\documentclass[10pt,twocolumn,twoside]{IEEEtran}
\ifCLASSOPTIONcompsoc \usepackage[caption=false,font=normalsize,labelfon
t=sf,textfont=sf]{subfig}
\else
\usepackage[caption=false,font=footnotesize]{subfig} 
\fi
\usepackage{fancyhdr}

\usepackage{amsmath,graphicx}
\usepackage[english]{babel}
\usepackage[latin1]{inputenc}
\usepackage{tikz}

\newcommand{\displaycomments}

\ifdefined\displaycomments
\newcommand{\note}[1]{\textcolor{red}{\emph{#1}}}
\else
\newcommand{\note}[1]{}
\fi

\newcommand{\displayold}

\ifdefined\displayold
\newcommand{\old}[1]{\textcolor{blue}{{#1}}}
\else
\newcommand{\old}[1]{}
\fi

\usepackage{xspace}
\usepackage{amsmath}
\usepackage{amsthm}
\usepackage{amssymb}
\usepackage{graphicx}      
\usepackage{algorithm}
\usepackage{algorithmicx} 
\usepackage{cite} 

\usetikzlibrary{arrows}
\usepackage{pgfplots}

\DeclareMathOperator*{\argmin}{argmin}

\newcommand{\iverson}[1]{[\![#1]\!]}
\newcommand{\E}{\mathop{\mathbb E}}

\renewcommand{\t}[1]{\mathrm{T}#1}
\newcommand{\N}{\mathcal{N}}

\renewcommand{\baselinestretch}{0.9} 
\usepackage{cite} 
\usepackage{amsmath}

\def\z{{\mathbf z}}

\renewcommand{\d}{\;\mathrm{d}}
\renewcommand{\t}{\mathrm{T}}

\newcommand{\tr}{\operatorname{tr}}

\newcommand{\var}{\operatorname{Var}}

\newtheorem{lemma}{Lemma}

\newcommand{\ST}{\operatorname{ST}}
\newcommand{\T}{\operatorname{T}}
\newcommand{\St}{\operatorname{t}}

\newcommand{\G}{\mathcal{G}}

\newcommand{\pluseq}{\stackrel{+}{=}}

\renewcommand{\d}[1]{\;\mathrm{d}#1}
\renewcommand{\t}[1]{\mathrm{T}#1}

\usepackage{booktabs}
\usepackage {ftnxtra}
\usepackage{algpseudocode}
\usepackage{epstopdf}

\newcommand{\figs}{figures/}

\newcommand{\eye}{I}

\newcommand{\zeros}{\mathrm{O}}

\newcommand{\STRTVBF}{STF\xspace}
\newcommand{\STRTVBS}{STS\xspace}
\newcommand{\recursive}{sequential\xspace}

\title{Skew-$t$ Filter and Smoother with\\Improved Covariance Matrix Approximation}

\author{Henri Nurminen, Tohid Ardeshiri, Robert Pich\'e,~\IEEEmembership{Senior Member,~IEEE}, and Fredrik Gustafsson,~\IEEEmembership{Fellow,~IEEE}
\thanks{H.\ Nurminen and R.\ Pich\'e are with the Laboratory of Automation and Hydraulic Engineering, Tampere University of Technology (TUT), PO Box 692, 33101 Tampere, Finland  (e-mails: henri.nurminen@here.com, robert.piche@tut.fi). H.\ Nurminen has received funding from TUT Graduate School, the Foundation of Nokia Corporation, Tekniikan edist\"amiss\"a\"ati\"o, and Emil Aaltonen Foundation. Henri Nurminen is currently with HERE Technologies Inc.}
\thanks{T.\ Ardeshiri is with the Division of Automatic Control, Department of Electrical Engineering, Link\"oping University, 58183, Link\"oping, Sweden and has received funding from Swedish research council (VR), project scalable Kalman filters for this work. 
T.\ Ardeshiri is currently with the Department of Engineering, University of Cambridge, Trumpington Street, Cambridge, CB2 1PZ, UK, (e-mail: ta417@cam.ac.uk).}
\thanks{F.\ Gustafsson is with the Division of Automatic Control, Department of Electrical Engineering, Link\"{o}ping University, 58183 Link\"{o}ping, Sweden, (e-mail: fredrik@isy.liu.se).}
}

\begin{document}
\def\figurename{Fig.}
\maketitle

\thispagestyle{fancy}
\renewcommand{\headrulewidth}{0pt}
\lfoot{\footnotesize \vspace{-9mm} \color{red} (c) 2018 IEEE. Personal use of this material is permitted. Permission from IEEE must be obtained for all other users, including reprinting/republishing this material for advertising or promotional purposes, creating new collective works for resale or redistribution to servers or lists, or reuse of any copyrighted components of this work in other works. Citation: H.\ Nurminen, T. Ardeshiri,  R.\ Pich\'e, and F.\ Gustafsson, ``Skew-$t$ Filter and Smoother with Improved Covariance Matrix Approximation'', \textit{IEEE Transactions on Signal Processing}, vol.\ 66, no.\ 21, pp.\ 5618--5633, 2018. DOI: 10.1109/TSP.2018.2865434}

\begin{abstract}
Filtering and smoothing algorithms for linear discrete-time state-space models with skew-$t$-distributed measurement noise are proposed. The algorithms use a variational Bayes based posterior approximation with coupled location and skewness variables to reduce the error caused by the variational approximation. Although the variational update is done suboptimally using an expectation propagation algorithm, our simulations show that the proposed method gives a more accurate approximation of the posterior covariance matrix than an earlier proposed variational algorithm. Consequently, the novel filter and smoother outperform the earlier proposed robust filter and smoother and other existing low-complexity alternatives in accuracy and speed. We present both simulations and tests based on real-world navigation data, in particular GPS data in an urban area, to demonstrate the performance of the novel methods. Moreover, the extension of the proposed algorithms to cover the case where the distribution of the measurement noise is multivariate skew-$t$ is outlined. Finally, the paper presents a study of theoretical performance bounds for the proposed algorithms.
\end{abstract}



\section{Introduction} \label{sec:introduction}

Asymmetric and heavy-tailed noise processes are present in many inference problems. In radio signal based distance estimation \cite{GusGun2005,BorsenChen2009,Kok2015}, for example, obstacles cause large positive errors that dominate over symmetrically distributed errors from other sources \cite{kaemarungsi2012}.  An example of this is the error histogram of time-of-flight in distance measurements collected in an indoor environment given in Fig.\ \ref{fig:error_hist}. The asymmetric distributions cannot be predicted by the normal or $t$-distributions that are equivalent in second order moments, because normal and $t$-distributions are symmetric distributions. The skew $t$-distribution \cite{branco2001,azzalini2003, gupta2003skew} is a generalization of the $t$-distribution that has the modeling flexibility to capture both skewness and heavy-tailedness of such noise processes. To illustrate this, Fig.~\ref{fig:2Diidcontours} shows the contours of the likelihood function for three range measurements where some of the measurements include large positive errors. In this example, skew-$t$, $t$, and normal measurement noise models are compared. Due to the additional modeling flexibility, the skew-$t$ based likelihood provides a more apposite spread of the probability mass than the normal and $t$ based likelihoods.

\begin{figure}
\centering
\includegraphics[width=0.98\columnwidth,trim=10mm 30mm 10mm 35mm,clip=true]{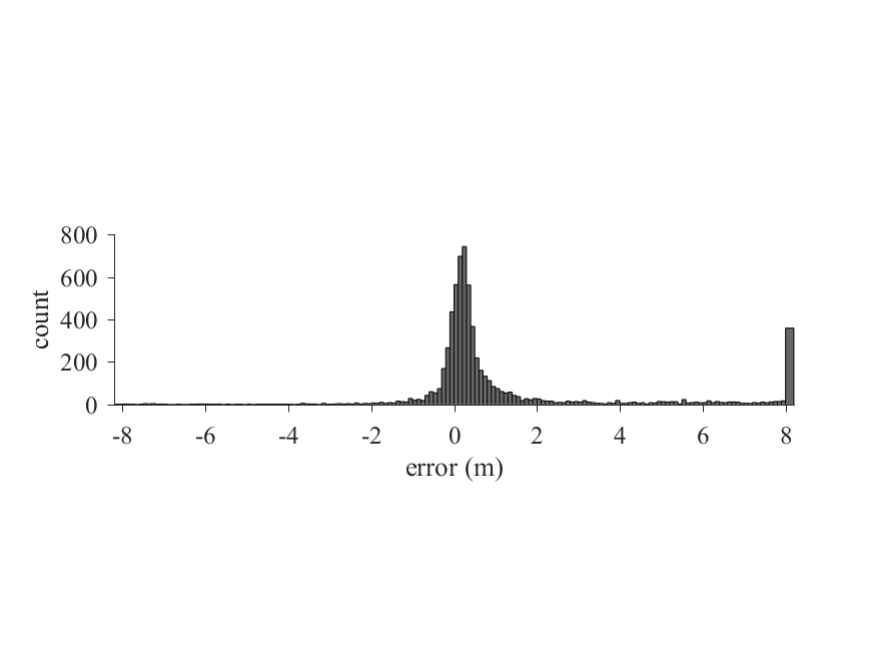}
\vspace{-3mm}
\caption{The error histogram in an ultra-wideband (UWB) ranging experiment described in \cite{nurminen2015b} shows positive skewness. The edge bars show the errors outside the figure limits.} \label{fig:error_hist}
\end{figure}
\begin{figure}
\centering
\hspace{-30mm}
\includegraphics[width=0.6\columnwidth]{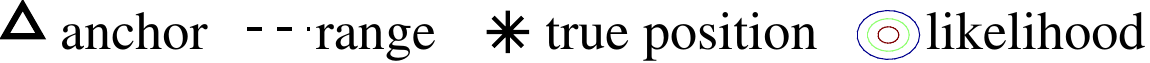}\\\vspace{1mm}
\includegraphics[trim=0mm 0mm 0mm 0mm, clip, width=0.31\columnwidth]{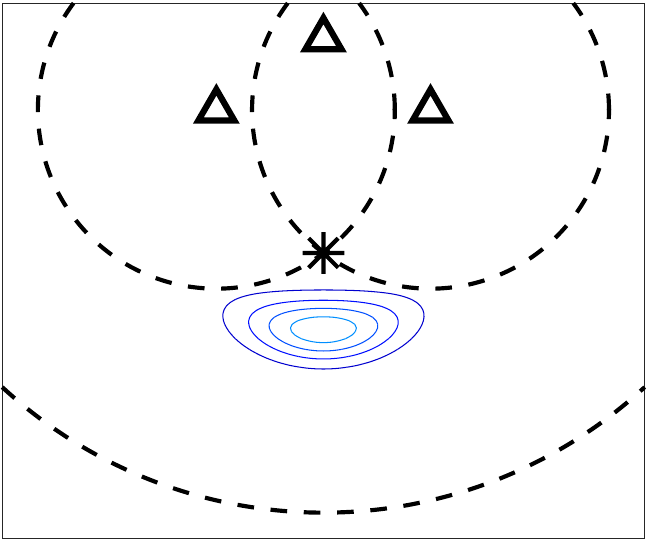}
\includegraphics[trim=0mm 0mm 0mm 0mm, clip, width=0.31\columnwidth]{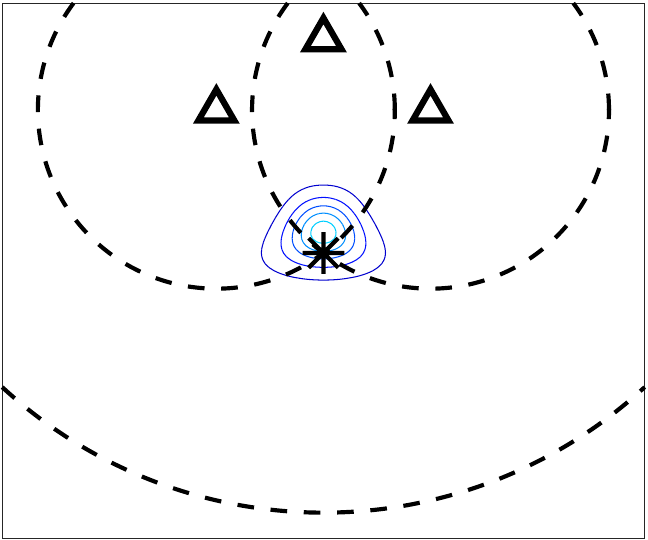}
\includegraphics[trim=0mm 0mm 0mm 0mm, clip, width=0.31\columnwidth]{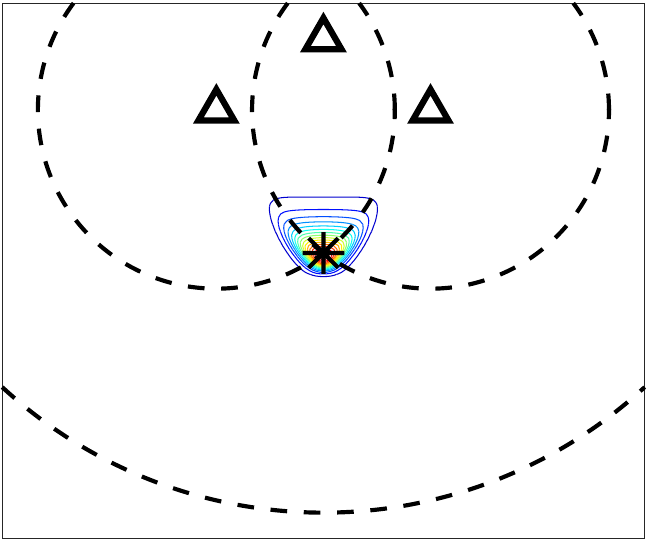}

\vspace{2mm}
\includegraphics[trim=0mm 0mm 0mm 0mm, clip, width=0.31\columnwidth]{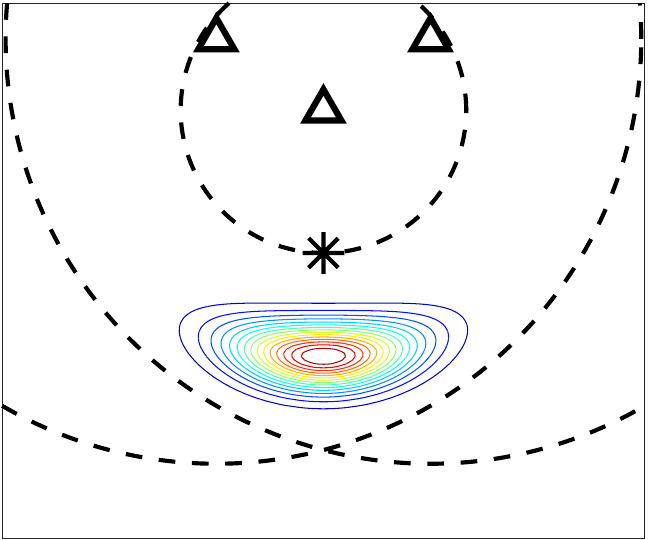}
\includegraphics[trim=0mm 0mm 0mm 0mm, clip, width=0.31\columnwidth]{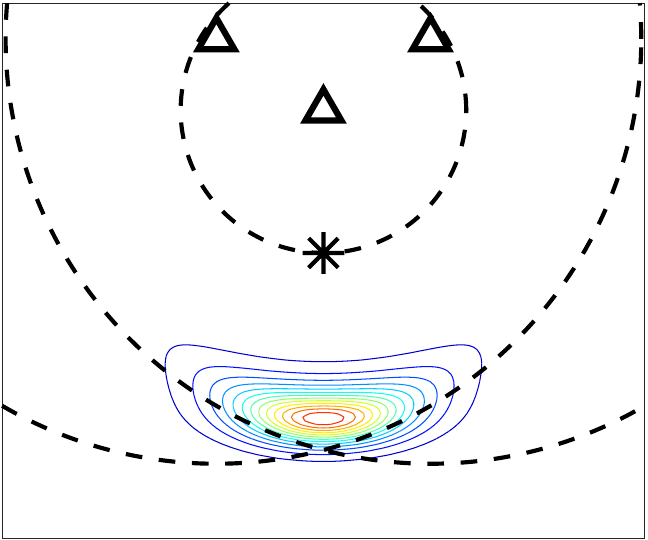}
\includegraphics[trim=0mm 0mm 0mm 0mm, clip, width=0.31\columnwidth]{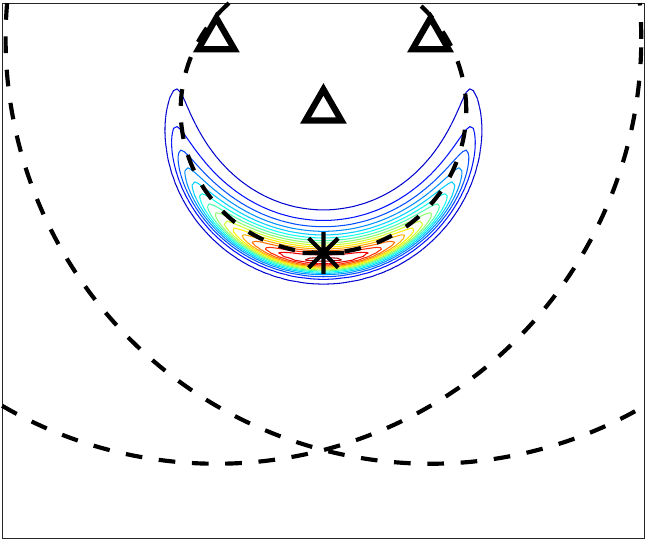}
\caption{The likelihood contours of distance measurements from three known anchors for the normal (left), $t$ (middle) and skew-$t$ (right) measurement noise models. The $t$ and skew-$t$ based likelihoods handle one large positive error (upper row), while only the skew-$t$ model handles the two large positive errors (bottom row) due to its asymmetry. The measurement model parameters are selected such that the degrees-of-freedom values and the first two moments coincide.} \label{fig:2Diidcontours}
\end{figure}

The applications of the skew distributions are not limited to radio signal based localization. In biostatistics skewed distributions are used as a modeling  tool for handling heterogeneous data involving asymmetric behaviors across subpopulations~\cite{fruhwirth2010}. 
In psychiatric research skew normal distribution is used to model asymmetric data~\cite{counsell2010}.  
Further, in economics skew normal and skew $t$-distributions are used as models for describing claims in property-liability insurance~\cite{eling2012}. 
 More examples describing approaches for analysis and modeling using multivariate skew normal and skew $t$-distributions in
econometrics and environmetrics are presented in~\cite{marchenko2010phd}. 

There are various algorithms dedicated to statistical inference of time series when the data exhibit asymmetric distribution. Particle filters \cite{doucet2000} can easily be adapted to skew noise distributions, but the computational complexity of these filters increases rapidly as the state dimension increases. A skew Kalman filter is proposed in \cite{naveau2005}, and in \cite{kim2014} this filter is extended to a robust scale-mixture filter using Monte Carlo integration. These solutions are based on state-space models where the measurement noise is a dependent process with skewed marginals. The article \cite{Rezaie2014} proposes filtering of independent skew measurement and process noises with the cost of increasing the filter state's dimension over time. In all the skew filters of \cite{naveau2005,kim2014,Rezaie2014}, sequential processing requires numerical evaluation of multidimensional integrals. The inference problem with skew likelihood distributions can also be cast into an optimization problem; \cite{Kok2015} proposes an approach to model the measurement noise in an ultra-wideband (UWB) based positioning problem using a tailored half-normal--half-Cauchy distribution. Skewness can also be modeled by a mixture of normal distributions (Gaussian mixtures, GM) \cite{GusGun2005}. There are many filtering algorithms for GM distributions such as Gaussian sum filter \cite{gsum1972} and interactive multiple model (IMM) filter \cite{bar1988tracking}. However, GMs have exponentially decaying tails and can thus be too sensitive to outlier measurements. Furthermore, in order to keep the computational cost of a Gaussian sum filter practicable, a mixture reduction algorithm (MRA) \cite{maybeck2006} is required, and these MRAs can be computationally expensive and involve approximations to the posterior density.

Variational Bayes (VB) method -based filtering and smoothing algorithms for linear discrete-time state-space models with skew-$t$ measurement noise are proposed in~\cite{nurminen2015a}. The VB approach avoids the increasing filter state dimensionality and numerical integrations by finding an optimal approximation with the constraint that the state is independent of the non-dynamic latent variables; this makes analytical marginalisation straightforward. To our knowledge, VB approximations have been applied to the skew $t$-distribution only in our earlier works \cite{nurminen2015a, nurminen2015b}, and by Wand et al.\ \cite{wand2011}, and the latter use a VB factorization different from ours and do not consider time-series inference. In tests with real UWB indoor localization data \cite{nurminen2015b}, this filter is shown to be accurate and computationally inexpensive.

This paper proposes improvements to the robust filter and smoother proposed in \cite{nurminen2015a}. 
Analogous to \cite{nurminen2015a}, the measurement noise is modeled by the skew $t$-distribution, and the proposed filter and smoother use VB approximations of the filtering and smoothing posteriors. However, the main contributions of this paper are (1) a novel VB factorization of the posterior and showing that at highly skewed models this factorization provides major improvement in both convergence speed of the VB iterations and accuracy of the estimate and covariance matrix, (2) an application of an existing expectation propagation (EP) based algorithm for approximating the statistics of a truncated multivariate normal distribution (TMND) that appears in the proposed VB algorithm, (3) a derivation of a greedy approach for a truncation ordering in the EP approximation of the TMND's moments, (4) a derivation of the Cram\'er--Rao lower bound (CRLB) for the proposed filter and smoother, and (5) the variational lower bound for the proposed VB factorization. A TMND is a multivariate normal distribution whose support is restricted (truncated) by linear constraints and that is re-normalized to integrate to unity. The aforementioned contributions improve the estimation performance of the skew-$t$ filter and smoother by reducing the covariance matrix underestimation common to many VB inference algorithms~\cite[Chapter 10]{Bishop2007}. This is shown by evaluating the proposed algorithms with both simulations and real-data tests in positioning using GNSS (global navigation satellite system) based pseudorange measurements. The tests show clear improvement in estimation accuracy compared to state-of-the-art low-complexity algorithms. In both simulations and real-data tests the proposed algorithms also outperform the earlier VB-based methods of \cite{nurminen2015a} in both estimation accuracy and speed of computations.

The rest of this paper is structured as follows. In Section~\ref{sec:problem_formulation}, the filtering and smoothing problem involving the univariate skew $t$-distribution is posed. In Section~\ref{sec:solution} a solution based on VB for the formulated problem is proposed. The proposed solution is evaluated in Sections~\ref{sec:simulations} and \ref{sec:realdata}. The essential expressions to extend the proposed filtering and smoothing algorithms to problems involving multivariate skew-$t$ (MVST) distribution are given in Section~\ref{sec:extension}. Performance bounds for time series data with MVST-distributed measurement noise are derived and evaluated in simulation in Section~\ref{sec:bound}. The concluding remarks are given in Section~\ref{sec:conclusions}.      


\section{Inference Problem formulation} \label{sec:problem_formulation}
Consider the linear and Gaussian state evolution model
\begin{subequations}
\label{eq:state-evo}
\begin{align}
p(x_1)&=\N(x_1;x_{1|0},P_{1|0}),\label{eq:prior} \\
x_{k+1}&=\ Ax_k+w_k, & w_k&\stackrel{\text{iid}}{\sim}\N(0,Q),
\end{align}
\end{subequations}
where 
 $\N(\cdot;\mu,\Sigma)$ denotes the probability density function (PDF) of the (multivariate) normal distribution with mean $\mu$ and covariance matrix $\Sigma$;
 $A \in \mathbb{R}^{n_x\times n_x}$ is the state transition matrix;
 $x_k\in\mathbb{R}^{n_x}$ indexed by $1\!\leq\! k\!\leq\! K$ is the state to be estimated with initial prior distribution \eqref{eq:prior}, 
 where the subscript ``$a|b$" is read ``at time $a$ using measurements up to time $b$"; and $w_k \in \mathbb{R}^{n_x}$ is the process noise. 
Further, the measurements $y_k\in\mathbb{R}^{n_y}$ are assumed to be governed by the measurement equation
\begin{align}
y_{k}&= C x_k+e_k, \label{eq:measurementmodel}
\end{align}
where $C \!\in\! \mathbb{R}^{n_y\times n_x}$ is the measurement matrix, and the measurement noise vector $e_k$ is independent of the process noise, and each component of $e_k$ follows an independent univariate skew $t$-distribution
\begin{align} \label{eq:measmod_independent}
[e_k]_i&\stackrel{\text{independent}}{\sim} \ST(0,R_{ii},\Delta_{ii},\nu_i),
\end{align}
where the operator $[\cdot]_i$ gives the $i$th entry of the argument vector, and $[\cdot]_{ij}$ gives the $(i,j)$ entry of its argument matrix. The model parameters can also be time-varying, but for the sake of lighter notation the $k$ subscripts on $A$, $Q$, $C$, $R$, $\Delta$, and $\nu$ are omitted. 
The univariate skew $t$-distribution $\ST(\mu,\sigma^2,\delta,\nu)$ is parametrized by its location parameter $\mu\in\mathbb{R}$, spread parameter $\sigma\in\mathbb{R}_+$, shape parameter $\delta \in\mathbb{R}$ and degrees of freedom $\nu\in\mathbb{R}_+$, and has the PDF
\begin{align}
\ST(z; \mu, \sigma^2,\delta,\nu)=2 \St(z;\mu,\sigma^2+\delta^2,\nu)\T(\widetilde{z};0,1,\nu+1)\label{eq:skewt},
\end{align}
where
\begin{align}
\St(z;\mu,\sigma^2,\nu)=\frac{\Gamma\left(\frac{\nu+1}{2}\right)}{\sigma\sqrt{\nu\pi}\Gamma\left(\frac{\nu}{2}\right)}\left(1+\frac{(z-\mu)^2}{\nu \sigma^2}\right)^{-\frac{\nu+1}{2}}
\end{align}
is the PDF of Student's $t$-distribution, $\Gamma(\cdot)$ is the gamma function, and 
$\widetilde{z} \!=\! \frac{(z-\mu)\delta}{\sigma} \left({\frac{\nu+1}{\nu(\sigma^2+\delta^2)+(z-\mu)^2}}\right)^{\frac{1}{2}}$.
Also, $\T(\cdot;0,1,\nu)$ denotes the cumulative distribution function (CDF) of Student's $t$-distribution with degrees of freedom $\nu$. Expressions for the first two moments of the univariate skew $t$-distribution can be found in \cite{lee2016}.

The model \eqref{eq:measmod_independent} with independent univariate skew-$t$-distributed measurement noise components is justified when one-dimensional noises of different sensors can be assumed to be statistically independent~\cite{nurminen2015a}. Extension and comparison to multivariate skew-$t$-distributed noise will be discussed in Section~\ref{sec:extension}.

The independent univariate skew-$t$ noise model~\eqref{eq:measmod_independent} induces the hierarchical representation of the measurement likelihood
\begin{subequations}
\label{eq:hierarchical}
\begin{align}
y_k|x_k,u_k,\Lambda_k &\thicksim \N(Cx_k+\Delta u_k,\Lambda_k^{-1}R),\\
u_k|\Lambda_k &\thicksim \N_+(0,\Lambda_k^{-1}),\\
[\Lambda_k]_{ii} &\thicksim \G(\tfrac{\nu_i}{2},\tfrac{\nu_i}{2}), \label{eq:hierarchical_lambda}
\end{align}
\end{subequations}
where
 $R\in\mathbb{R}^{n_y\times n_y}$ is a diagonal matrix whose diagonal elements' square roots $\sqrt{R_{ii}}$ are the spread parameters of the skew $t$-distribution in
 \eqref{eq:measmod_independent};
 $\Delta\in\mathbb{R}^{n_y\times n_y}$ is a diagonal matrix whose diagonal elements $\Delta_{ii}$ are the shape parameters;
 $\nu\in\mathbb{R}_+^{n_y}$ is a vector whose elements $\nu_i$ are the degrees of freedom.
 $\Lambda_k$ is a diagonal matrix with a priori independent random diagonal elements $[\Lambda_k]_{ii}$. Also, $\N_+(\mu,\Sigma)$ is the TMND with closed positive orthant as support, location parameter $\mu$, and squared-scale matrix $\Sigma$. Furthermore, $\G(\alpha, \beta)$ is the gamma distribution with shape parameter $\alpha$ and rate parameter $\beta$.

Bayesian smoothing means finding the smoothing posterior $\pi(x_{1:K},u_{1:K},\Lambda_{1:K})\!\triangleq\!p(x_{1:K},u_{1:K},\Lambda_{1:K}|y_{1:K})$. In \cite{nurminen2015a}, the smoothing posterior is approximated by a factorized distribution of the form $q_{\text{\cite{nurminen2015a}}}(x_{1:K},u_{1:K},\Lambda_{1:K}) \!\triangleq\! q_x(x_{1:K})\,q_u(u_{1:K})\,q_\Lambda(\Lambda_{1:K})$. Subsequently, the approximate posterior distributions are computed using the VB approach. The VB approach minimizes the Kullback--Leibler divergence (KLD) $D_{\text{KL}}(q||p)\!\triangleq\!\int q(x)\log\frac{q(x)}{p(x)} \mathrm{d} x$~\cite{CoverT2006} of the true posterior from the factorized approximation. 
That is, $D_{\text{KL}}(q_{\text{\cite{nurminen2015a}}}||\pi)$ is minimized in \cite{nurminen2015a}. An approximate Bayesian filter update, i.e.\ an approximation the filtering posterior $p(x_k,u_k,\Lambda_k|y_{1:k})$ given a normal filtering prior for $x_k$, is then a smoother update with $K\!=\!1$.

The numerical simulations in \cite{nurminen2015a} manifest covariance matrix underestimation, which is a known weakness of the VB approach \cite[Chapter 10]{Bishop2007}. One of the contributions of this paper is to reduce the covariance underestimation of the filter and smoother proposed in \cite{nurminen2015a} by removing independence approximations from the VB factorization. The proposed filter and smoother are presented in Section~\ref{sec:solution}.


\section{Proposed Filter and Smoother} \label{sec:solution}
\subsection{VB factorization}
Using Bayes' theorem, the state evolution model~\eqref{eq:state-evo}, and the likelihood \eqref{eq:hierarchical},  the joint smoothing posterior PDF is
 \begin{align}
&\pi(x_{1:K},u_{1:K},\Lambda_{1:K}) \nonumber\\
&\propto \N(x_1;x_{1|0},P_{1|0}) \prod_{l=1}^{K-1}\N(x_{l+1};Ax_l,Q) \nonumber\\
&\hspace{2mm}\times\prod_{k=1}^K  \N(y_k;Cx_k+\Delta u_k,\Lambda_k^{-1}R)\,\N_+(u_k; 0,\Lambda_k^{-1})\nonumber\\
&\hspace{2mm}\times\prod_{k=1}^K\prod_{i=1}^{n_y}\G\left([\Lambda_k]_{ii};\frac{\nu_i}{2},\frac{\nu_i}{2}\right).
\end{align}
The posterior is not analytically tractable. We propose to seek an approximation in the form
\begin{align}
\label{eq:factors}
\pi(x_{1:K},&u_{1:K},\Lambda_{1:K})\approx \hat{q}_{xu}(x_{1:K},u_{1:K})\,\hat{q}_\Lambda(\Lambda_{1:K}) ,
\end{align}
where the factors in \eqref{eq:factors} are specified by
\begin{align}
&\hat{q}_{xu},\hat{q}_{\Lambda}=\argmin_{{q}_{xu},{q}_{\Lambda}}D_{\text{KL}}(q_\text{N}\,||\,\pi),
\end{align}
where $q_\text{N}(x_{1:K},u_{1:K},\Lambda_{1:K})\!\triangleq\! q_{xu}(x_{1:K},u_{1:K})\,q_\Lambda(\Lambda_{1:K})$ is the factorized approximation. Hence, $x_{1:K}$ and $u_{1:K}$ are not approximated as independent as in \cite{nurminen2015a} because they can be highly correlated \textit{a posteriori} \cite{nurminen2015a}. The analytical solutions for $\hat{q}_{xu}$ and $\hat{q}_\Lambda$ are obtained by cyclic iteration of
\begin{subequations}
\label{eqn:IterativeOptimization}
\begin{align}
&\log {q}_{xu}(x_{1:K},u_{1:K}) \!\leftarrow\! \E_{{q}_{\Lambda}}[\log p(y_{1:K},x_{1:K},u_{1:K},\Lambda_{1:K})]\!+\!c_{xu}\label{eqn:IterativeOptimizationxu}\\
&\log {q}_{\Lambda}(\Lambda_{1:K}) \!\leftarrow\! \E_{{q}_{xu}}[\log p(y_{1:K},x_{1:K},u_{1:K},\Lambda_{1:K})]\!+\!c_{\Lambda}\label{eqn:IterativeOptimizationL}
\end{align}
\end{subequations}\normalsize
where $\leftarrow$ is the assignment or reassignment operator, and the expected values on the right hand sides are taken with respect to the current $q_{xu}$ and $q_\Lambda$~\cite[Chapter 10]{Bishop2007}\cite{TzikasLG2008,Beal03}. Also, $c_{xu}$ and $c_\Lambda$  are constants with respect to the variables $(x_{1:K},u_{1:K})$ and $\Lambda_{1:K}$, respectively.

The detailed derivation of the proposed smoother is given in Appendix~\ref{sec:smoother}. The distribution $q_{xu}(x_{1:K},u_{1:K})$ is a $K\!\times\!(n_x\!+\!n_y)$-dimensional TMND, where the underlying normal distribution can be obtained using the Rauch--Tung--Striebel smoother (RTSS) \cite{RTS-1965}. However, the first two moments of each $x_k$-marginal are required in the computation of the expectation in~\eqref{eqn:IterativeOptimizationL}, and a TMND's moments cannot be computed in closed form. This renders the smoother impractical, since there is no efficient algorithm for approximating the moments of a large TMND. To obtain a practical smoother algorithm, we replace the RTSS's forward filtering step with the assumed-normal filter where each joint filtering distribution of $x_k$ and $u_k$, each of them being a TMND, is approximated by a normal distribution with the matched mean and covariance matrix. Because each of these filtering distributions is a low-dimensional ($(n_x\!+\!n_y)$-dimensional) TMND, their means and covariance matrices can be approximated efficiently using the computationally light algorithm discussed in Subsection \ref{sec:tmnd_moments}. The result of the assumed-normal filter is then fed into the standard RTSS's backward smoothing step. The obtained skew-$t$ smoother (STS) algorithm is given in Algorithm~\ref{table:smoothing}.

In short, one iteration of the proposed smoother consists of a forward filtering step for the variables $(x_k,u_k)$, of a standard RTSS backward smoothing step for the same variables, and of updating $q_\Lambda(\Lambda_k)$ based on the residuals and covariance matrices of each $q(x_k,u_k)$. The forward filtering step is done with a KF-type algorithm where each filtering distribution is modified with the approximative TMND's moments formula.

An approximative filtering update step can be derived as the smoother for a state-space model with just one time-instant. Because each $q_{xu}(x_k,u_k)$ distribution is again a low-dimensional TMND, the moments of each $q_{xu}(x_k,u_k)$ can be approximated quickly. By approximating the $x_k$-marginal $\int q_{xu}(x_k,u_k) \mathrm{d}u_k$ of the final VB iteration's TMND with a normal distribution, we obtain a recursive filtering algorithm, the skew-$t$ filter (\STRTVBF) of Algorithm~\ref{table:filtering}. While the marginal $\int q_{xu}(x_k,u_k) \mathrm{d}u_k$ is not exactly normal but consists of non-truncated components of a TMND, it is unimodal and has $\mathbb{R}^{n_x}$ as support, so the normal distribution with the matching first and second moments is a standard approximation. This normality approximation does not affect the convergence of the filtering VB iterations, but there is no convergence proof for the VB iterations when the moments of the TMND are approximated. However, the approximative VB iterations show better accuracy and convergence speed in the numerical simulations presented in Sections \ref{sec:simulations} than the VB iterations with the factorization $q_{[20]}$.

In short, one VB iteration in the proposed filter's measurement update step consists of updating $q(x_k,u_k)$ with a KF update, modifying its joint mean and covariance matrix with the approximative TMND's moments formulas, and finally updating $q_\Lambda(\Lambda_k)$ based on the residual and covariance matrix of $q(x_k,u_k)$.

We propose three stopping criteria for the VB iterations of the filter and smoother: small enough change in the estimate, small enough increase in the variational lower bound (practical only for the filter), and a fixed number of iterations. The computation of the variational lower bound is explained in Subsection \ref{sec:var_lower_bound}. In our tests we fix the number of VB iterations to five, because we found that the estimation accuracy does not improve after five iterations. Fixing the number of VB iterations is the most practical option in terms of predictability of the computation times, but the required number of iterations has to be verified for each model specifically.

\begin{algorithm}[t]
\caption{Smoothing for skew-$t$ measurement noise}\label{table:smoothing}
\small
\newcommand{\mlambda}[1]{\Lambda_{#1|K}}
\begin{algorithmic}[1]
\State \textbf{Inputs:} $A$, $C$, $Q$, $R$, $\Delta$,  $\nu$, $x_{1|0}$,  $P_{1|0}$, $y_{1:K}$, \texttt{APPROX\_TMND}
\State $\mlambda{k} \gets I_{n_y}$ for $k=1\cdots K$, $A_z \gets  \left[\begin{smallmatrix}A&0\\0&0\end{smallmatrix}\right]$, $C_z \gets \left[\begin{smallmatrix}C&\Delta\end{smallmatrix}\right]$
\Repeat
\Statex \hspace{2mm}\textit{update $q_{xu}(x_{1:K},u_{1:K})$}
	\For{$k$ = 1 to $K$}
		\State $Z_{k|k-1} \gets \mathrm{blockdiagonal}(P_{k|k-1}, \mlambda{k}^{-1})$
		\State $K_{z} \gets Z_{k|k-1} C_z^\t (CP_{k|k-1}C^\t\!+\!\Delta\mlambda{k}^{-1}\Delta^\t\!+\!\mlambda{k}^{-1}R)^{-1}$
		\State $\widetilde{z}_{k|k} \gets \left[\begin{smallmatrix} x_{k|k-1}\\0 \end{smallmatrix}\right]+K_z(y_k-Cx_{k|k-1})  $
		\State $\widetilde{Z}_{k|k}\gets (I-K_z C_z)P_{k|k-1}$
		\State $[z_{k|k}, \!Z_{k|k}] \!\!\gets\!\! \texttt{APPROX\_TMND}(\widetilde{z}_{k|k},\!\widetilde{Z}_{k|k},$ \!\tiny$\{n_x\!+\!1 \!\cdots\! n_x\!+\!n_y\}$ \!\!\!\small$)$
		\Statex \vspace{-1.15ex}
		\State $x_{k|k} \gets [z_{k|k}]_{1:n_x},\ P_{k|k} \gets [Z_{k|k}]_{1:n_x,1:n_x}$
    		\Statex \vspace{-1.1ex}
	 	\State $x_{k+1|k} \gets Ax_{k|k}$
	  	\State $P_{k+1|k} \gets AP_{k|k}A^\t+Q$	
	\EndFor
	\For{$k$ = $K-1$ down to $1$ }
		\State $G_k\gets Z_{k|k} A_z Z_{k+1|k}^{-1}$
		\State $z_{k|K}\gets z_{k|k}+G_k(z_{k+1|K}-A_z z_{k|k})$
		\State $Z_{k|K}\gets Z_{k|k}+G_k(Z_{k+1|K}-Z_{k+1|k})G_k^\t$
		\Statex \vspace{-1.1ex}
		\State $x_{k|K} \gets [z_{k|K}]_{1:n_x},\ P_{k|K} \gets [Z_{k|K}]_{1:n_x,1:n_x}$
		\State $u_{k|K} \!\!\!\gets\!\!\! [z_{k|K}]_{n_x+(1:n_y)}, U_{k|K} \!\!\!\gets\!\!\! [Z_{k|K}]_{n_x+(1:n_y),n_x+(1:n_y)}$
	\EndFor
	\vspace{0.8mm}
	\Statex \hspace{2mm}\textit{update $q_\Lambda(\Lambda_{1:K})$}
	\For{$k$ = $1$ to $K$}
		\State $\Psi \gets (y_k-C_z z_{k|K})(y_k-C_z z_{k|K})^\t R^{-1}$
		\Statex $\hspace{1.5cm}  + C_z Z_{k|K} C_z^\t R^{-1} +u_{k|K}u_{k|K}^\t + U_{k|K}$
		\Statex \vspace{-1.1ex}
		\State \textbf{for} $i$ = 1 to $n_y$ \textbf{do} $[\mlambda{k}]_{ii} \gets\frac{\nu_i+2}{\nu_i+\Psi_{ii}}$ \textbf{end for} \label{listing:lambda_smoothing}
	\EndFor	
\Until{\textbf{converged}}
\State \textbf{Outputs: $x_{k|K}$ and  $P_{k|K}$ for $k=1\cdots K$ } 
\end{algorithmic}
\end{algorithm}
\begin{algorithm}[t]
\caption{Filtering for skew-$t$ measurement noise}\label{table:filtering}
\small
\newcommand{\mlambda}[1]{\Lambda_{#1|k}}
\begin{algorithmic}[1]
\State \textbf{Inputs:} $A$, $C$, $Q$, $R$, $\Delta$,  $\nu$, $x_{1|0}$,  $P_{1|0}$, $y_{1:K}$, \texttt{APPROX\_TMND}
\State $\Lambda \gets \eye_{n_y}$,\quad$C_z \gets \left[\begin{smallmatrix}C&\Delta\end{smallmatrix}\right]$
\For{$k$ = 1 to $K$}
	\State $[a_{k|k}]_i \gets \tfrac{\nu_i+2}{2}$, $[b_{k|k}]_i \gets \tfrac{\nu_i+2}{2}$ for $i=1,\cdots,n_y$
	\Repeat
		\State $[\Lambda_{k|k}]_{ii} \gets \tfrac{[a_{k|k}]_i}{[b_{k|k}]_i}$ for $i=1,\cdots,n_y$
		\vspace{0.8mm}
		\Statex \hspace{6mm}\textit{update $q_{xu}(x_k,u_k)$}
\State $Z_{k|k-1} \gets \mathrm{blockdiagonal}(P_{k|k-1}, \mlambda{k}^{-1})$
		\State $K_{z} \gets Z_{k|k-1} C_z^\t (CP_{k|k-1}C^\t\!+\!\Delta\mlambda{k}^{-1}\Delta^\t\!+\!\mlambda{k}^{-1}R)^{-1}$
		\State $\widetilde{z}_{k|k} \gets \left[\begin{smallmatrix} x_{k|k-1}\\0 \end{smallmatrix}\right]+K_z(y_k-Cx_{k|k-1})  $
		\State $\widetilde{Z}_{k|k}\gets (I-K_z C_z)P_{k|k-1}$
		\State $[z_{k|k}, \!Z_{k|k}] \!\!\gets\!\! \texttt{APPROX\_TMND}(\widetilde{z}_{k|k},\!\widetilde{Z}_{k|k},$ \!\tiny$\{n_x\!+\!1 \!\cdots\! n_x\!+\!n_y\}$ \!\!\!\small$)$
		\Statex \vspace{-1.15ex}
		\State $x_{k|k} \gets [z_{k|k}]_{1:n_x},\ P_{k|k} \gets [Z_{k|k}]_{1:n_x,1:n_x}$
		\State $u_{k|k} \!\!\gets\!\! [z_{k|k}]_{n_x+(1:n_y)}, U_{k|k} \!\!\gets\!\! [Z_{k|k}]_{n_x+(1:n_y),n_x+(1:n_y)}$
		\vspace{0.8mm}
		\Statex \hspace{6mm}\textit{update $q_\Lambda(\Lambda_k)$}
		\State $\Psi \gets (y_k-C_z z_{k|k})(y_k-C_z z_{k|k})^\t R^{-1}$
		\Statex $\hspace{1.5cm}  +C_z Z_{k|k} C_z^\t R^{-1}+u_{k|k}u_{k|k}^\t + U_{k|k}$
		\Statex \vspace{-1.1ex}
		\State \textbf{for} $i$ = 1 to $n_y$ \textbf{do} $[b_{k|k}]_{i} \gets\frac{\nu_i+\Psi_{ii}}{2}$ \textbf{end for} \label{listing:lambda_filtering}
		\Until{\textbf{converged}}
	\State $x_{k+1|k} \gets Ax_{k|k}$
	\State $P_{k+1|k} \gets AP_{k|k}A^\t+Q$	
\EndFor
\State \textbf{Outputs: $x_{k|k}$ and  $P_{k|k}$ for $k=1\cdots K$ } 
\end{algorithmic}
\end{algorithm}

\subsection{TMND's moments} \label{sec:tmnd_moments}

The mean and covariance matrix of a TMND can be computed using the formulas presented in \cite{tallis1961}. They require evaluating the CDFs of general multivariate normal distributions. The \textsc{Matlab} function \verb+mvncdf+ implements the numerical quadrature  of \cite{genz2004} in 2 and 3 dimensional cases and the quasi-Monte Carlo method of \cite{Genz02} for the dimensionalities 4--25. However, these methods can be prohibitively slow. Therefore, we approximate the TMND's moments using a fast sequential algorithm that is based on the expectation propagation (EP) algorithm \cite{minka2001}. An EP algorithm for computing the mean, covariance matrix, and the truncated probability of a TMND is derived in \cite{cunningham2013_arxiv}. The method is initialized with the original normal density whose parameters are then updated by applying one linear constraint at a time. For each constraint, the mean and covariance matrix of the once-truncated normal distribution are computed analytically, and the once-truncated distribution is approximated by a non-truncated normal with the updated moments. The EP is an iterative algorithm, so each truncation can be re-made when, roughly speaking, the effect of the previous iteration of the considered truncation is removed from the normal distribution's moments. One iteration of this method is illustrated in Fig.~\ref{fig:truncation}, where a bivariate normal distribution truncated into the positive quadrant is approximated with a non-truncated normal distribution.
\begin{figure*}
\centering
\subfloat[]{\includegraphics[width=0.15\textwidth]{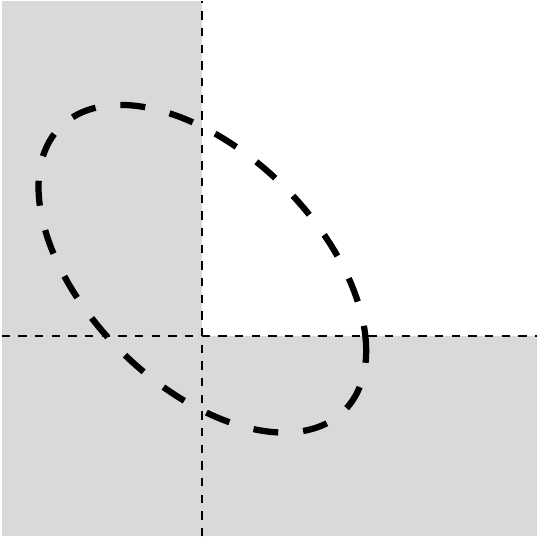}}\ \label{fig:re?ctrunc1}
\subfloat[]{\includegraphics[width=0.15\textwidth]{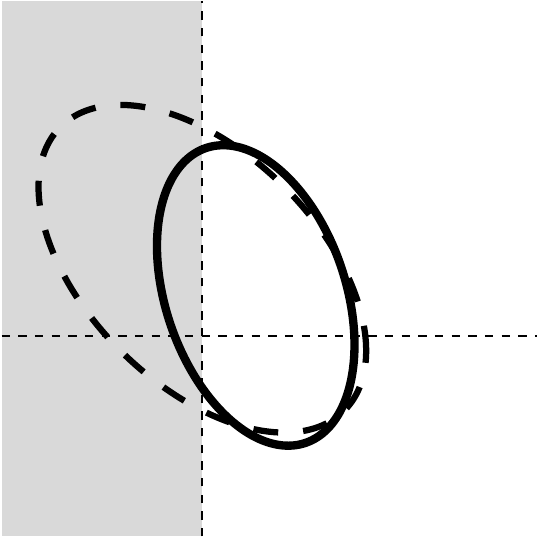}}\ \label{fig:rectrunc2}
\subfloat[]{\includegraphics[width=0.15\textwidth]{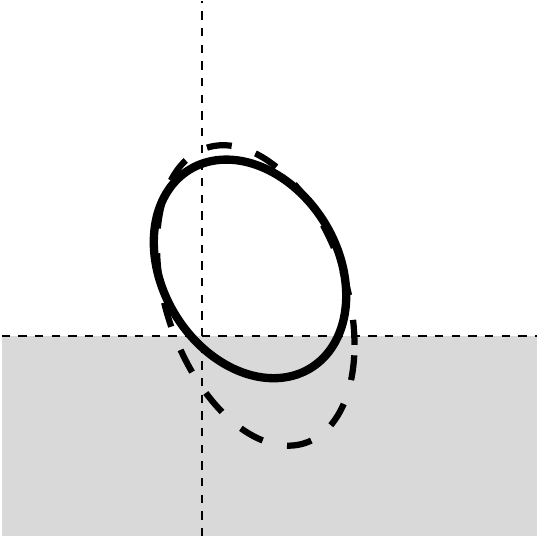}}\ \label{fig:rectrunc3}
\subfloat[]{\includegraphics[width=0.15\textwidth]{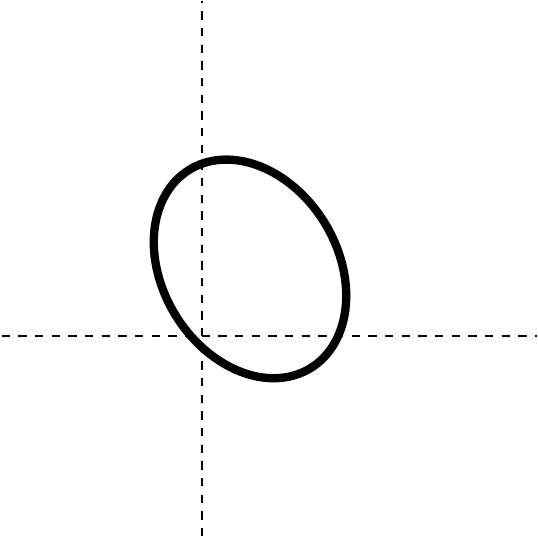}} \label{fig:rectrunc4} 
\includegraphics[width=0.22\textwidth,trim=0mm -78mm 0mm 0mm,clip]{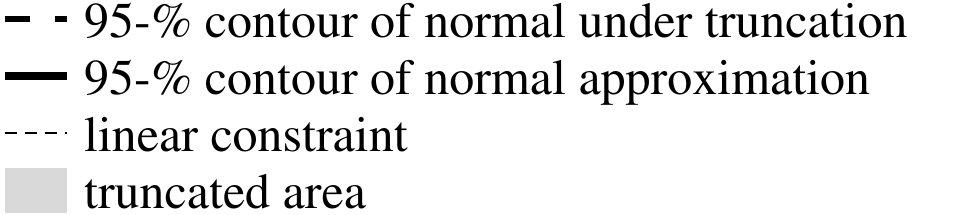}
\hfill
\caption{An iteration of the EP algorithm for approximating a truncated normal distribution with a normal distribution: (a) the original normal distribution's contour ellipse that contains 95\,\% of the probability, and the truncated area in gray, (b) the first applied truncation in gray, and the 95-\% contour of the resulting normal approximation, (c) the second applied truncation in gray, and the 95-\% contour of the normal approximation, (d) the final normal approximation.}
\label{fig:truncation}
\end{figure*}

The result of the EP algorithm depends on the order in which the constraints are applied. Finding the optimal order of applying the truncations is a problem that has combinatorial complexity. Hence, we adopt a greedy approach, whereby the constraint to be applied is chosen from among the remaining constraints so that the resulting once-truncated normal distribution is closest to the true TMND. By Lemma~\ref{lem:optimal-seq}, the optimal constraint to select is the one that truncates the most probability. The optimality is with respect to a KLD as the measure. For example, in Fig.~\ref{fig:truncation} the vertical constraint truncates more probability, so it is applied first.
\begin{lemma} \label{lem:optimal-seq}
Let $p(\z)$ be a {TMND} with the support $\{\z \geq 0\}$ and $q(\z)=\N(\z;\mu,\Sigma)$. Then,
\begin{align}
\argmin_i D_{\normalfont{\text{KL}}}\!\left(p(\z)\,\big| \big|\,\tfrac{1}{c_i}q(\z)\iverson{\z_i \geq 0}\right) = \argmin_i \tfrac{\mu_i}{\sqrt{\Sigma_{ii}}} ,
\end{align} 
where $\mu_i$ is the $i$th element of $\mu$, $\Sigma_{ii}$ is the $i$th diagonal element of $\Sigma$, $\iverson{\cdot}$ is the Iverson bracket, and $c_i\!=\! \int{q(\z)\iverson{\z_i \geq 0}}\d\z$.
\begin{align}
\text{Proof: }& D_{\normalfont{\text{KL}}}\!\left(p(\z)\,\big| \big|\,\tfrac{1}{c_i}q(\z)\iverson{\z_i \geq 0}\right)\nonumber\\
&\pluseq -\int{p(\z)\log(\tfrac{1}{c_i}q(\z)\iverson{\z_i \geq 0})\d\z}\\
&= \log{c_i} - \int{p(\z)\log q(\z)\d\z} \pluseq \log{c_i} ,
\end{align} 
where $\pluseq$ means equality up to an additive constant. 
Since  $c_i$ is an increasing function of $\frac{\mu_i}{\sqrt{\Sigma_{ii}}}$, the proof follows. \vspace{-2ex} \begin{flushright}$\blacksquare$\end{flushright}
\end{lemma}

The obtained EP algorithm with the greedy processing sequence for computing the mean and covariance matrix of a given multivariate normal distribution truncated to the positive orthant is given in Algorithm~\ref{table:recursive}. The algorithm can also give the logarithm of the positive orthant's probability $\alpha$, which is required in computing the variational lower bound. In many programming languages a numerically robust method to implement the line \ref{listing:pperc} of the algorithm in Algorithm~\ref{table:recursive} is using the scaled complementary error function $\mathrm{erfcx}$ through
\begin{equation}
\frac{\phi(\xi)}{\Phi(\xi)} = \frac{\sqrt{2/\pi}}{\mathrm{erfcx}(-\xi/\sqrt{2})} .
\end{equation}

Unfortunately, the EP algorithm does not in general admit guaranteed convergence or error bounds. However, Cunningham et al.\ \cite{cunningham2013_arxiv} present an extensive Monte Carlo study on the performance of the EP in approximating the truncated probability of a TMND, showing that this EP algorithm is reliable provided that there are no redundant truncating constraints and that the support of the distribution is hyperrectangular. Cunningham et al.\ also imply that the same result applies to approximating the moments of the TMND. These conditions are fulfilled in our case, as all the truncating hyperplanes are non-redundant and aligned with coordinate axes.
\begin{algorithm}[t]
\caption{Greedy expectation propagation for the moments of normal distribution truncated to positive orthant\newline(function $[\mu,\Sigma,\alpha] \gets \texttt{APPROX\_TMND}(\mu,\Sigma,\mathcal{T}$))}\label{table:recursive}
\small
\begin{algorithmic}[1]
\State \textbf{Inputs:} $\mu$, $\Sigma$, and set of the truncated components' indices $\mathcal{T}$
\State $\widetilde\mu \gets \mu$,\quad$\widetilde\Sigma \gets \Sigma$
\State $\alpha \gets  - \frac{1}{2} \widetilde\mu^\t \widetilde\Sigma^{-1} \widetilde\mu$,\quad$M \gets I_{n_\mu}$
\State $\tau_k\gets0$,\quad$\eta_k\gets0$ for $k=1,2,\ldots,n_\mu$.
\Repeat
\State $\mathcal{T}' \gets \mathcal{T}$
\While{$\mathcal{T}' \neq \emptyset$}
	\State $k \gets \argmin_i\{ \mu_i/\sqrt{\Sigma_{ii}} \mid i \in \mathcal{T}' \}$
	\State $s^2 \gets 1/(1/\Sigma_{kk}-\tau_k)$
	\State $m \gets s^2 (\mu_k/\Sigma_{kk}-\eta_k)$
	\State $\xi \gets m / s$
	\State $\epsilon \gets \phi(\xi) / \Phi(\xi)$ \Comment{$\phi$ is the PDF of $\N(0,1)$, $\Phi$ its CDF}  \label{listing:pperc}
	\State $\overline{m} \gets m + \epsilon s$
	\State $\overline{s}^2 \gets (1-\xi\epsilon-\epsilon^2) s^2$
	\State $\overline\tau_k \gets 1/\overline{s}^2-1/s^2-\tau_k$,\quad$\tau_k\gets \tau_k+\overline\tau_k$
	\State $\overline\eta_k \gets \overline{m}/\overline{s}^2-m/s^2-\eta_k$,\quad$\eta_k\gets \eta_k+\overline\eta_k$
	\State $\mu \gets \mu + \frac{\overline\eta_k-\overline\tau_k \mu_k}{1+\overline\tau_k \Sigma_{kk}} \cdot \Sigma_{:,k}$ \Comment{mean update}
	\State $\Sigma \gets \Sigma - \frac{\overline\tau_k}{1+\overline\tau_k \Sigma_{jj}} \cdot \Sigma_{:,k}\Sigma_{k,:}$ \Comment{covariance update}
	\State $M \gets M + \tau_k L_{k,:}^\t L_{k,:}$ \Comment{$LL^\t\!=\!\widetilde\Sigma$}
	\State $\alpha \gets \alpha + \log\big(\Phi(\xi)\big) + \frac{1}{2}\log(1+\tau_k s^2) + \frac{1}{2} \tau_k \mu_k^2$
	\State \hspace{5mm} $+ \frac{1}{2} \frac{m^2\tau_k - 2m\eta_k - s^2 \eta_k^2}{1+\tau_k s^2}$ \Comment{log-probability update}
	\State $\mathcal{T}' \gets \mathcal{T}' \backslash \{k\}$
\EndWhile
\State \!$\alpha \!\gets\! \alpha \!-\! \frac{1}{2} \log(\det(M)) \!+\! \frac{1}{2} \mu^\t \widetilde\Sigma^{-1} \mu$
\Until{\textbf{converged}}
\State \textbf{Outputs:} moments $\mu$, $\Sigma$, and the logarithm of the positive orthant's probability $\alpha$
\end{algorithmic}
\end{algorithm}

\subsection{Variational lower bound} \label{sec:var_lower_bound}
When the PDF $p(x|y)$ is approximated with the PDF $q(x)$, the variational lower bound is
\begin{align}
\mathcal{L}(q) =& \int q(x) \log\frac{p(y,x)}{q(x)} \mathrm{d}x.
\end{align}
Minimizing the KLD is equivalent to maximizing the variational lower bound \cite[Ch.\ 21]{murphy2012}. Therefore, the variational lower bound can be used as a debugging means and convergence criterion for the VB iterations because the lower bound should increase at each iteration. Furthermore, because the logarithmic marginal likelihood $\log p(y)$ is the sum of the variational lower bound and the KLD, the maximal variational lower bound can be used as an approximation for $\log p(y)$. The model evidence in Bayesian comparison can thus be approximated with $\exp(\mathcal{L}(q))$ \cite[Ch.\ 21.5.1.6]{murphy2012}.

When evaluated immediately after the VB filter update of $q_{xu}(x_k,u_k)$, the variational lower bound for the skew-$t$ filter is
\begin{align} \label{eq:lb}
\mathcal{L}_\text{f}(q)
=&  \log\N\big(y; Cx_{k|k-1}, CP_{k|k-1}C^\t \!+\! \Delta\Lambda_{k|k}^{-1}\Delta^\t \!+\! \Lambda_{k|k}^{-1}R \big) \nonumber\\
 +& \sum_{j=1}^{n_y} \bigg[ [a_{k|k}]_j \bigg( 1 + \log\left( \tfrac{[a_{k|k}]_j-1}{[b_{k|k}]_j} \right) - \tfrac{[a_{k|k}]_j-1}{[b_{k|k}]_j} \bigg) \nonumber\\
 -& \log\left(\tfrac{[a_{k|k}]_j}{[b_{k|k}]_j}\right) \bigg] + n_y \log(2) + \log\alpha_{k|k},
\end{align}
where the notations follow those in Algorithm \ref{table:filtering}, and $\alpha_{k|k}$ is the probability of the positive orthant for the distribution $\N([\widetilde{z}_{k|k}]_{n_x+(1:n_y)},[\widetilde{Z}_{k|k}]_{n_x+(1:n_y),n_x+(1:n_y)})$. The probability $\alpha_{k|k}$ can be computed using the EP algorithm in Algorithm \ref{table:recursive}. The derivation of the lower bound \eqref{eq:lb} is straightforward but tedious and omitted here. Unfortunately, evaluation of the variational lower bound for the smoother is impractical because its expression includes a probability of the positive orthant given a high-dimensional normal distribution.


\section{Simulations} \label{sec:simulations}

Our numerical simulations use satellite navigation pseudorange measurements and the model
\begin{equation} \label{eq:measmodel}
[y_k]_i \!=\! \left\| s_i \!-\! [x_k]_{1:3} \right\| + [x_k]_4 + [e_k]_i , [e_k]_i \stackrel{\text{iid}}{\sim} \mathrm{ST}(0,1\,\text{m},\delta\,\text{m},4)
\end{equation}
where $s_i \!\in\! \mathbb{R}^3$ is the $i$th satellite's position, $[x_k]_4 \!\in\! \mathbb{R}$ is bias with prior $\N(0,(0.75\,\text{m})^2)$, and $\delta \!\in\! \mathbb{R}$ is a parameter. The model is linearized using the first order Taylor polynomial approximation, and the linearization error is negligible because the satellites are far relative to the magnitude of uncertainty in the prior. The satellite constellation of the Global Positioning System (GPS) from the first second of the year 2015 provided by the International GNSS Service \cite{dow2009} is used with 8 visible satellites. The root-mean-square error (RMSE) is computed for the position $[x_k]_{1:3}$ as
\begin{equation}
\text{RMSE} = \sqrt{\frac{1}{K}\sum_{k=1}^K \big\| [x_{k|k}]_{1:3} - [x_k]_{1:3} \big\|^2},
\end{equation}
where $x_{k|k}$ is the filter estimate and $x_k$ is the true state. The computations are made with \textsc{Matlab}.

\subsection{Computation of TMND statistics}

In this subsection we study the computation of the moments of the untruncated components of a TMND. For each Monte Carlo replication, one state value is generated from the prior $x\!\sim\!\N(0,\mathrm{diag}(20^2,20^2,0.22^2,0.1^2)\,\text{m}^2)$, and one measurement vector is generated from the model \eqref{eq:measmodel} with  $\nu\!=\!\infty$ degrees of freedom (corresponding to skew-normal likelihood). 10\,000 Monte Carlo replications are used. The compared methods are expectation propagation (EP) with the greedy truncation order and one, two, three, four, and five  EP iterations (GEP1, GEP2, GEP3, GEP4, GEP5), the variational Bayes (VB), and the analytical formulas of \cite{tallis1961} using \textsc{Matlab} function \verb+mvncdf+ (MVNCDF). VB is an update of the skew $t$ VB filter (STVBF) \cite{nurminen2015a} where the heavy-tailedness variable $\overline{\Lambda}_1$ is fixed to identity $\eye_{n_y}$ and the VB iteration is terminated when the position estimate changes less than 0.005\,m or at the 1000th iteration. The reference solution for the expectation value is an importance sampling (IS) update with 50\,000 samples and the prior as the importance distribution.

Fig.\ \ref{fig:mvncdftest} shows the distributions of the estimates' differences from the IS estimate. The errors are given per cent of the IS's estimation error. The box levels are 5\,\%, 25\,\%, 50\,\%, 75\,\%, and 95\,\% quantiles and the asterisks show minimum and maximum values. The results indicate that the accuracy of the EP approximation of the mean does not improve after two EP iterations. MVNCDF is slightly more accurate than GEP2 in the cases with high skewness, but MVNCDF's computational load is roughly 40\,000 times that of the GEP2. This justifies the use of the EP approximation.

The approximation of the posterior covariance matrix is tested by studying the normalized estimation error squared (NEES) values \cite[Ch.\ 5.4.2]{bar-shalom}
\begin{equation}
\text{NEES}_k = (x_{k|k}-x_k)^\t P_{k|k}^{-1} (x_{k|k}-x_k) ,
\end{equation}
where $x_{k|k}$ and $P_{k|k}$ are the filter's output mean and covariance matrix, and $x_k$ is the true state. The algorithms' $\text{NEES}_1$ values averaged over the Monte Carlo replications are given in Table \ref{tab:nees}. If the covariance matrix is correct, the $\text{NEES}_1$ is $\chi^2$-distributed with 3 degrees of freedom because the position is 3-dimensional, so the nominal expected value is 3 \cite[Ch.\ 5.4.2]{bar-shalom}. VB shows large average $\text{NEES}_1$ values when $\delta$ is large, which indicates that VB underestimates the covariance matrix. Apart from MVNCDF, the GEP algorithms show average $\text{NEES}_1$ values closest to 3, so the EP provides a more accurate covariance matrix approximation than VB. Indicated by average $\text{NEES}_1$ being slightly smaller than 3, GEP1 in fact overestimates the covariance matrix when $\delta$ is large, but this issue is mostly fixed by the second EP iteration.
\begin{figure}[t]
\centering
\includegraphics[width=0.98\columnwidth]{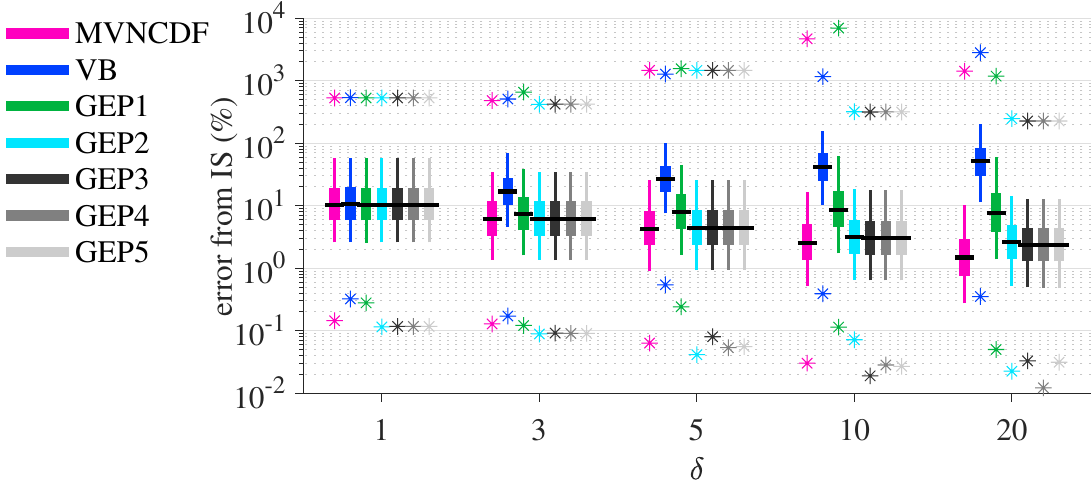} \hfill
\vspace{-3mm}
\caption{Two EP iterations suffice. MVNCDF is slightly more accurate than the proposed GEP but computationally heavy.} \label{fig:mvncdftest}
\end{figure}
\begin{table}[t]
\renewcommand{\arraystretch}{1.3}
\centering
\caption{The average $\text{NEES}_1$ values. GEP1's average $\text{NEES}_1$ is closer to the optimal value 3 than that of VB, so EP gives a more accurate posterior covariance matrix.}
\label{tab:nees}
\vspace{-2mm}
\begin{tabular}{c|cccccccc}
$\delta$ & 1 & 3 & 5 & 10 & 20 \\
\hline
MVNCDF & 3.0 & 3.0 & 3.0 & 3.0 & 3.0 \\
VB & 3.8 & 9.1 & 19.1 & 65.6 & 229.2 \\
GEP1 & 3.0 & 2.9 & 2.8 & 2.7 & 2.7 \\
GEP2 & 3.0 & 3.0 & 3.0 & 2.9 & 2.9 \\
GEP3 & 3.0 & 3.0 & 3.0 & 2.9 & 2.9 \\
GEP4 & 3.0 & 3.0 & 3.0 & 2.9 & 2.9 \\
GEP5 & 3.0 & 3.0 & 3.0 & 2.9 & 2.9
\end{tabular}
\end{table}

The order of the truncations in the EP algorithm affects the performance only when there are clear differences in the amounts of probability mass under each truncation. We compare GEP1 with the EP iteration with a random truncation order (REP1). In REP1 any of the non-optimal constraints is chosen randomly at each truncation. Fig.\ \ref{fig:rtrandtest} presents an example where $\delta\!=\!20$, and the measurement noise realization $e$ has been generated from the skew normal distribution and then modified by
\begin{equation}
e_j = \min\{\min\{e_{1:n_y}\},0\} - c\sqrt{1+20^2} ,
\end{equation}
where $j$ is a random index, and $c$ is a parameter. A large $c$ generates one negative outlier to each measurement vector, which results in one truncation with significantly larger truncated probability mass than the rest of the truncations. Fig.\ \ref{fig:rtrandtest} shows the percentual difference of REP1 error from GEP1 error; i.e.\ a positive difference means that GEP1 is more accurate. The errors here refer to distance from the IS estimate. The figure shows that with large $c$ GEP1 is more accurate than REP1. Thus, the effect of the truncation ordering on the accuracy of the EP approximation is more pronounced when there is one truncation that truncates much more than the rest. This justifies our greedy approach and the result of Lemma~\ref{lem:optimal-seq} for ordering the truncations.
\begin{figure}
\centering
\includegraphics[width=0.8\columnwidth]{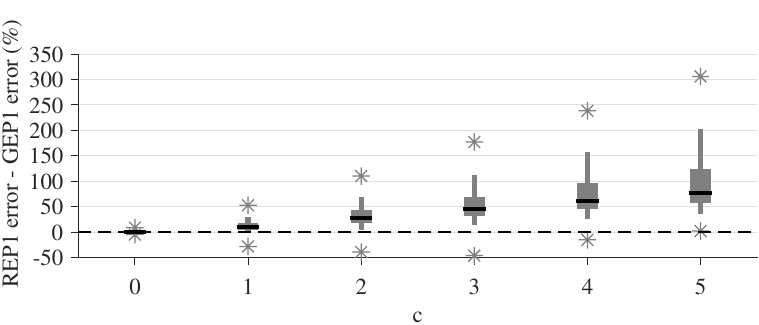}
\vspace{-3mm}
\caption{The proposed GEP1 outperforms REP1 when one negative outlier is added to the measurement noise vector because there is one truncation that truncates much more probability than the rest.}
\label{fig:rtrandtest}
\end{figure}

The skew-$t$ measurement model essentially implies that given the scaling variable $\Lambda$, we are observing the sum $Cx+\Delta u$ plus normally distributed noise. Fig.\ \ref{fig:tmnd_approx_example} compares the EP approximation and the 30-iteration VB approximation of the posterior distribution for the model
\begin{subequations}\label{eq:tmnd_approx_example}
\begin{align}
&p(x,u) = \N(x;0,1) \cdot \N_+(u;0,1)\\
&p(y|x,u) = \N(y; x+\delta u,0.1^2)
\end{align}
\end{subequations}
with the measurement value $y\!=\!1$ and with $\delta$ values 0.1, 0.5, and 1. Fig.\ \ref{fig:tmnd_approx_example} illustrates that when $\delta$ is large, $x$ and $u$ are highly correlated. This makes VB seriously underestimate the covariance matrix, and EP provides a better approximation of the joint posterior and the marginal posterior of $x$.
\begin{figure}
\centering
\begin{minipage}{0.32\columnwidth}
\includegraphics[width=\columnwidth,trim=0mm 0mm 0mm 0mm,clip=true]{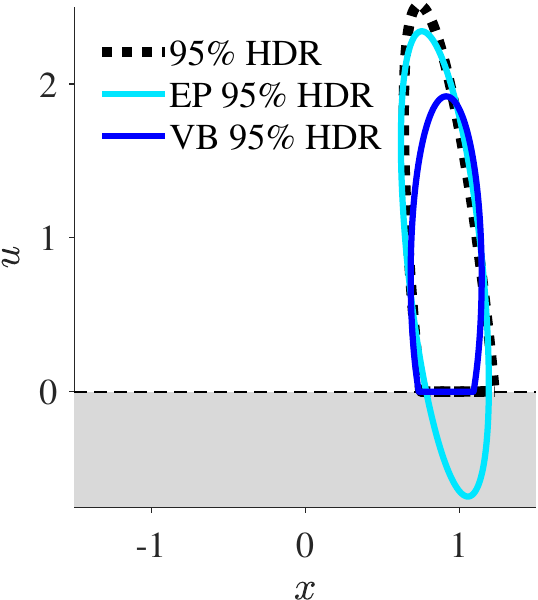}
\subfloat[]{
\includegraphics[width=\columnwidth,trim=2mm 0mm 0mm 0mm,clip=true]{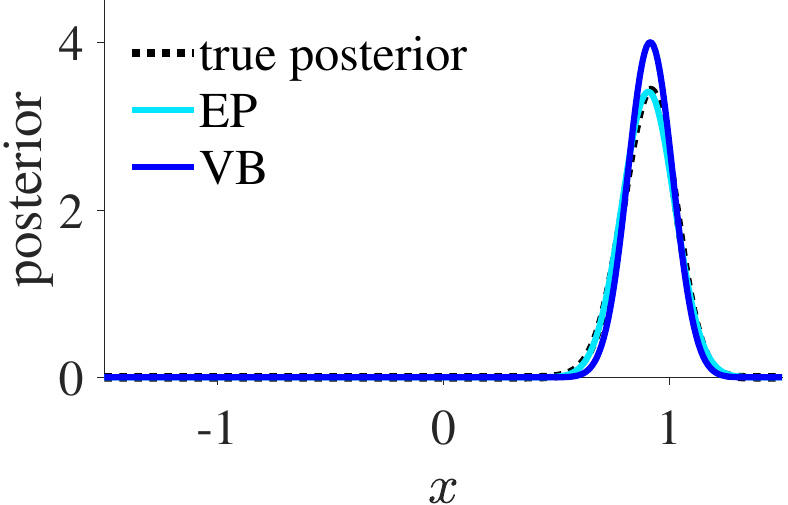}
}
\end{minipage}
\hfill
\begin{minipage}{0.32\columnwidth}
\includegraphics[width=\columnwidth,trim=0mm 0mm 0mm 0mm,clip=true]{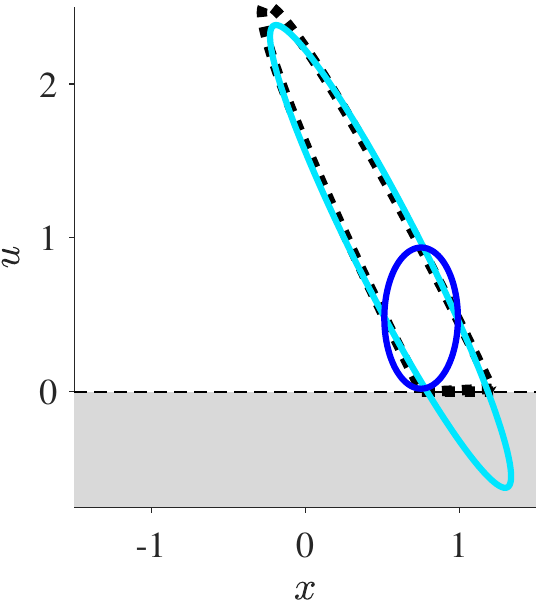}
\subfloat[]{
\includegraphics[width=\columnwidth,trim=2mm 0mm 0mm 0mm,clip=true]{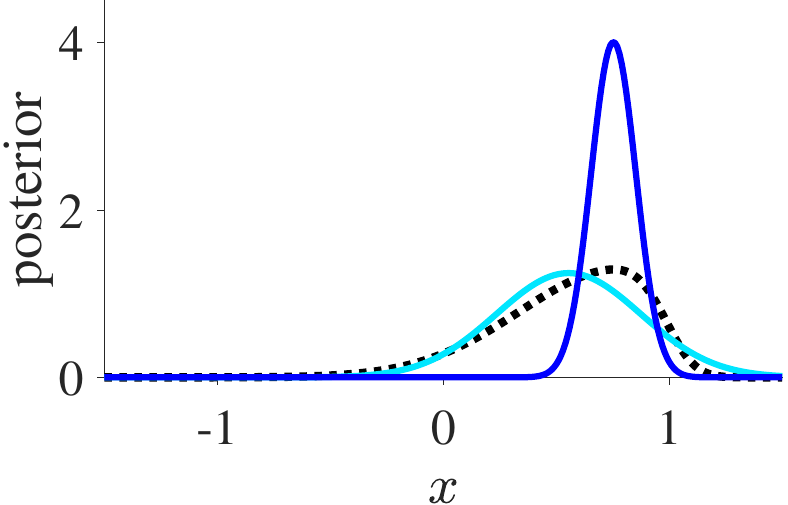}
}
\end{minipage}
\hfill
\begin{minipage}{0.32\columnwidth}
\includegraphics[width=\columnwidth,trim=0mm 0mm 0mm 0mm,clip=true]{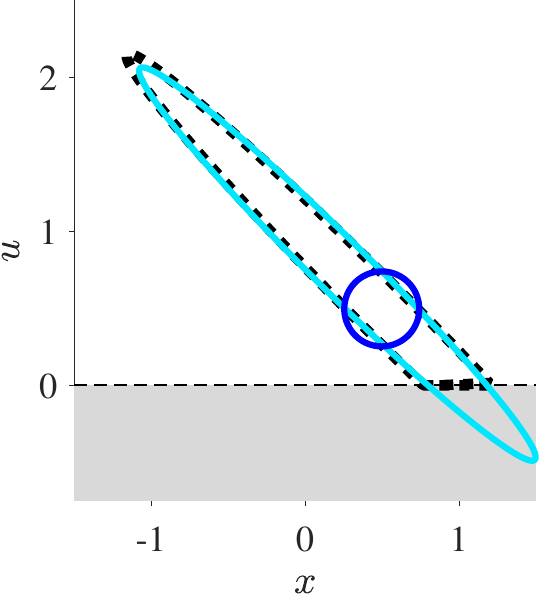}
\subfloat[]{
\includegraphics[width=\columnwidth,trim=2mm 0mm 0mm 0mm,clip=true]{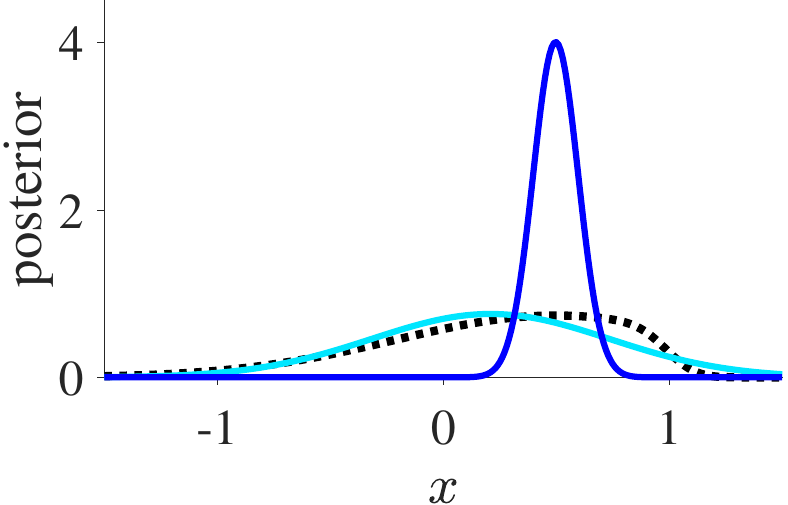}
}
\end{minipage}
\vspace{-1mm}
\caption{EP gives a better approximation than VB for a bivariate normal distribution of $(x,u)$, where $u$ is truncated to be positive. The figures show the 95\,\% high-density regions (HDR) of the posteriors $p(x,u|y\!=\!1)$ (upper row) and the marginal posteriors $p(x|y\!=\!1)$ (lower row) of the model \eqref{eq:tmnd_approx_example}. (a) $\delta\!=\!0.1$, (b) $\delta\!=\!0.5$, (c) $\delta\!=\!1$.}\label{fig:tmnd_approx_example}
\end{figure}

\subsection{Skew-$t$ inference}

In this section, the proposed skew-$t$ filter (\STRTVBF) is compared with state-of-the-art filters using numerical simulations of a 100-step trajectory. The tested \STRTVBF uses two EP iterations. The measurement model is given in \eqref{eq:measmodel}, and the state evolution model is a random walk with process noise covariance $Q\! =\! \mathrm{diag}(q^2, q^2, 0.2^2, 0)\,\text{m}$, where $q$ is a parameter. The compared methods are a bootstrap-type PF, STVBF \cite{nurminen2015a}, $t$ variational Bayes filter (TVBF) \cite{piche2012}, and Kalman filter (KF) with measurement validation gating \cite[Ch.\ 5.7.2]{bar-shalom} that discards the measurement components whose normalized innovation squared is larger than the $\chi_1^2$-distribution's 99\,\% quantile. The used KF parameters are the mean and variance of the used skew $t$-distribution, and the TVBF parameters are obtained by matching the degrees of freedom with that of the skew $t$-distribution and computing the maximum likelihood location and scale parameters for a set of pseudo-random numbers generated from the skew $t$-distribution. The results are based on 10\,000 Monte Carlo replications.

Fig.\ \ref{fig:time_vs_acc} illustrates the filter iterations' convergence when the measurement noise components $[e_k]_i$ in \eqref{eq:measmodel} are generated independently from the univariate skew $t$-distribution. The figure shows that the proposed \STRTVBF's median RMSE does not improve after five VB iterations, and \STRTVBF outperforms the other filters in RMSE already with two VB iterations, except for PF that is the minimum-RMSE solution. Furthermore, Fig.\ \ref{fig:time_vs_acc} shows that \STRTVBF's converged state is close to the PF's converged state in RMSE, and PF can require as many as 10\,000 particles to outperform \STRTVBF. In our implementation, the PF with 10\,000 particles is computationally roughly 15 times heavier that the \STRTVBF with five VB iterations. \STRTVBF also converges faster than STVBF when the process noise variance parameter $q$ is large.

Fig.\ \ref{fig:rmse_skewt} shows the distributions of the RMSE differences from the \STRTVBF's RMSE as percentages of the \STRTVBF's RMSE. STF1 is the skew-$t$ filter with just one EP iteration per a VB iteration. STF, STF1, and TVBF use five VB iterations, and STVBF uses 30 VB iterations. STF clearly has the smallest RMSE when $\delta\!\geq\!3$, i.e.\ when the skewness is high. STF1 and STF (with 2 EP iterations) have similar accuracies, so one EP iteration may be sufficient in practice. Unlike STVBF, the new \STRTVBF improves accuracy even with small $q$ and large $\delta$, which can be explained by the improved covariance matrix approximation.

The proposed smoother is also tested with the measurement model \eqref{eq:measmodel} and the random-walk state model. The compared smoothers are the proposed skew-$t$ smoother with two EP iterations (\STRTVBS), skew-$t$ variational Bayes smoother (STVBS) \cite{nurminen2015a}, $t$ variational Bayes smoother (TVBS) \cite{piche2012}, and the RTSS with 99\,\% measurement validation gating \cite{RTS-1965}. Fig.\ \ref{fig:time_vs_acc_smoother} shows that \STRTVBS has lower RMSE than the smoothers based on symmetric distributions. Furthermore, \STRTVBS's VB iteration converges in five iterations or less, so it is faster than STVBS.
\begin{figure}[t]
\centering
\includegraphics[width=0.75\columnwidth]{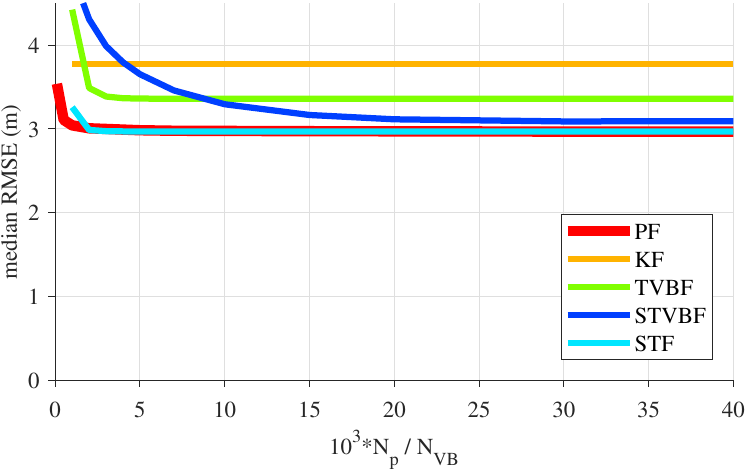}
\vspace{-3mm}
\caption{The proposed \STRTVBF's median RMSE does not improve after $N_\text{VB}\!=\!5$ VB iterations per time instant. The required number of PF particles $N_\text{p}$ can be more than 10\,000, and STVBF \cite{nurminen2015a} can require 30 VB iterations. The $x$-axis is $10^3\!\cdot\! N_\text{p}$ for PF and $N_\text{VB}$ for the rest of the filters. $q\!=\!5,\ \delta\!=\!5$.}
\label{fig:time_vs_acc}
\end{figure}
\begin{figure}[t]
\centering
\includegraphics[clip,width=0.48\columnwidth]{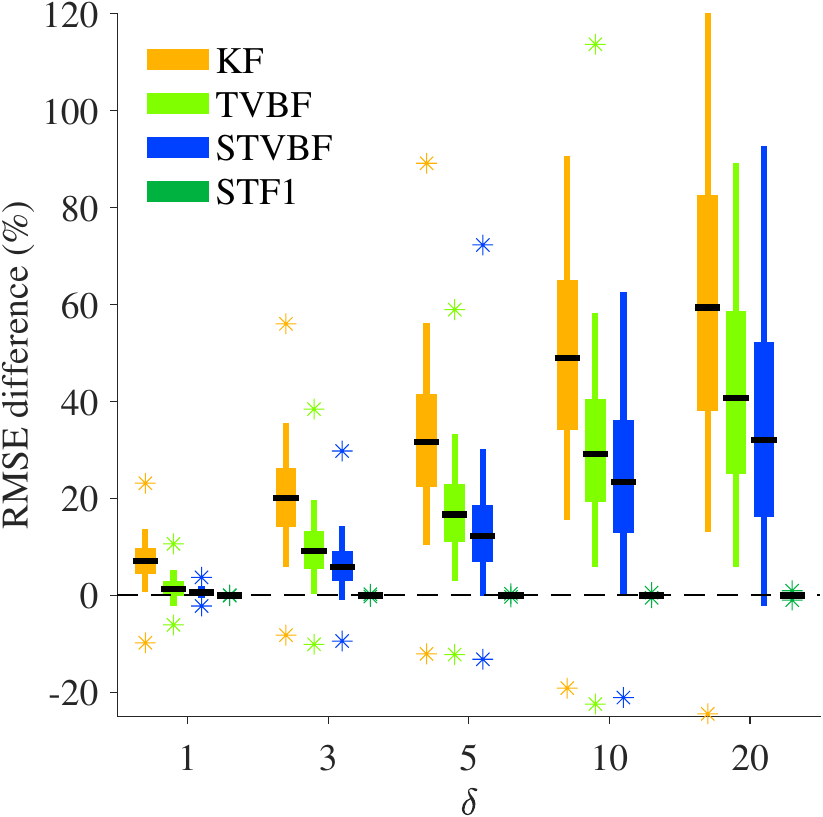}
\includegraphics[width=0.48\columnwidth]{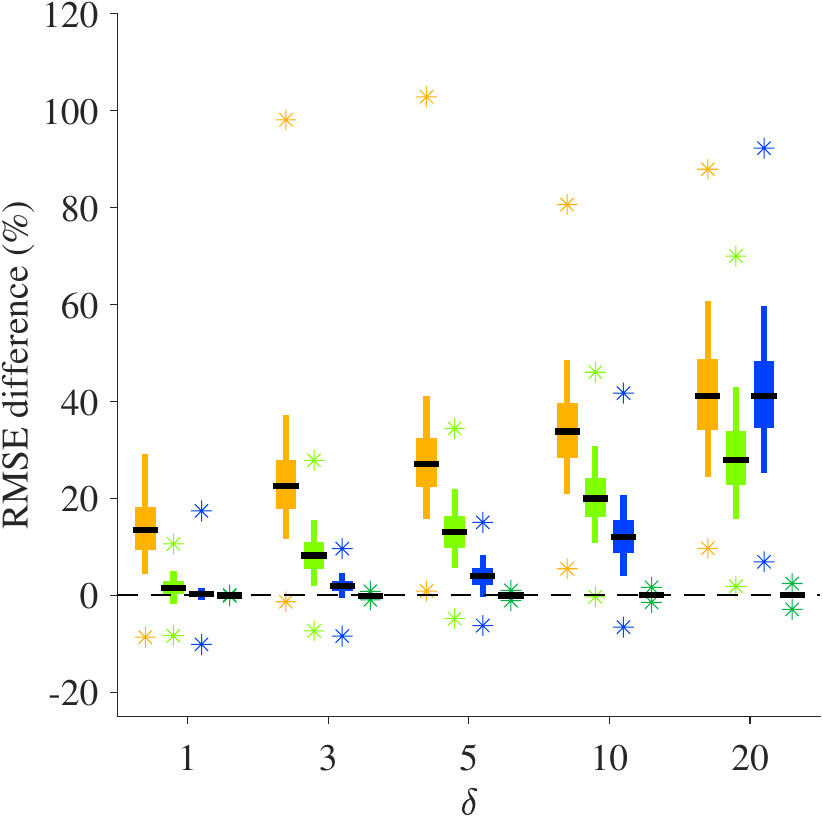}
\vspace{-2mm}
\caption{The proposed \STRTVBF outperforms the comparison methods with skew-$t$-distributed noise. RMSE differences from \STRTVBF's RMSE per cent of the \STRTVBF's RMSE. The difference to STVBF \cite{nurminen2015a} increases as skewness increases or when process noise variance reduces. (left) $q=0.5$, (right) $q=5$.} \label{fig:rmse_skewt}
\vspace{2ex}
\end{figure}
\begin{figure}[b]
\centering
\includegraphics[width=0.75\columnwidth]{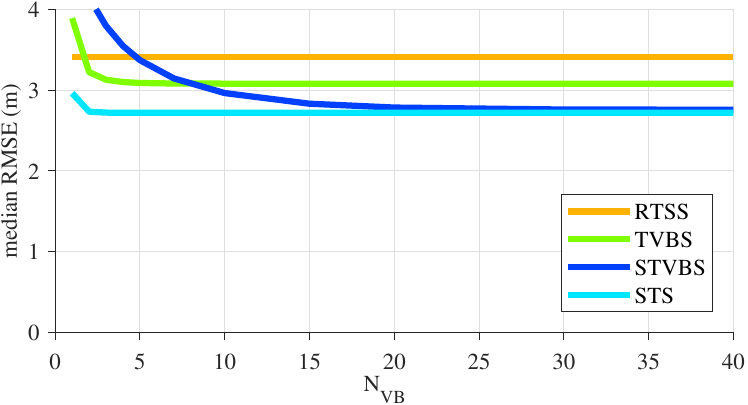}
\vspace{-3mm}
\caption{Five \STRTVBS  iterations give the converged state's RMSE, while STVBS \cite{nurminen2015a} can require 30 iterations. $q\!=\!5,\ \delta\!=\!5$.}
\label{fig:time_vs_acc_smoother}
\end{figure}


\section{Tests with real data}\label{sec:realdata}

\newcommand{\mut}{\mu_\text{t}}
\newcommand{\st}{\sigma_\text{t}}
\newcommand{\nut}{\nu_\text{t}}
\newcommand{\mun}{\mu_\text{n}}
\newcommand{\sn}{\sigma_\text{n}}

Two GNSS positioning data sets were collected in central London (UK) to test the filters' performance in a challenging real-world satellite positioning scenario with numerous non-line-of-sight measurements. The data include time-of-flight based pseudorange measurements from GPS satellites. Each set contains a trajectory that was collected by car using a u-blox M8 GNSS receiver. The lengths of the tracks are about 8\,km and 10\,km, the durations are about an hour for each, and measurements are received at about one-second intervals. The first track is used for fitting the filter parameters, while the second track is used for studying the filters' positioning errors. A ground truth was measured using an Applanix POS-LV220 system that improves the GNSS solution with tactical grade inertial measurement units. The GPS satellites' locations were obtained from the broadcast ephemerides provided by the International GNSS Service \cite{dow2009}. The algorithms are computed with \textsc{Matlab}.

In this test, both the user position $l_k\in\mathbb{R}^3$ and the receiver clock error $b_k\in\mathbb{R}$ follow the almost-constant velocity model used in \cite[Section IV]{axelrad1994}. Thus, the filter state being $x_k\!=\!\left[\begin{smallmatrix} l_k^\t&\dot{l}_k^\t&b_k&\dot{b}_k\end{smallmatrix}\right]^\t \in \mathbb{R}^8$, the state evolution model is
\begin{equation}
x_{k+1} = \begin{bmatrix} \eye_3 & d_k \eye_3 & \zeros_{3\times2} \\ \zeros_3 & \eye_3 & \zeros_{3\times2} \\ \zeros_{2\times3} & \zeros_{2\times3} & \left[\begin{smallmatrix} 1&d_k\\0&1 \end{smallmatrix} \right] \end{bmatrix} x_k + w_k ,
\end{equation}
where
\[
w_k \sim \N\left(0, \begin{bmatrix} \tfrac{q^2 d_k^3}{3} \eye_3 & \tfrac{q^2 d_k^2}{2} \eye_3 & \zeros_{3\times2} \\ \tfrac{q^2 d_k^2}{2} \eye_3 & q^2 d_k \eye_3 & \zeros_{3\times2} \\ \zeros_{2\times3} & \zeros_{2\times3} & \left[\begin{smallmatrix} s_b d_k+\tfrac{s_f d_k^3}{3}&\tfrac{s_f d_k^2}{2}\\\tfrac{s_f d_k^2}{2}&s_f d_k \end{smallmatrix} \right]  \end{bmatrix} \right) ,
\]
and $d_k$ is the time difference of the measurements in seconds. The used parameter values are $q=0.5\,\mathrm{m}/\mathrm{s}^\frac{3}{2}$, $s_b=70\,\frac{\mathrm{m}^2}{\mathrm{s}}$, and $s_f=0.6\,\tfrac{\mathrm{m}^2}{\mathrm{s}^3}$. The initial prior is a normal distribution with mean given by the Gauss--Newton method with the first measurement and a large covariance matrix.

The measurement model is the same pseudorange model that is used in the simulations of Section \ref{sec:simulations}, i.e.\
\begin{equation}
[y_k]_i = \| s_{i,k}-[x_k]_{1:3}\| + [x_k]_7 + [e_k]_i ,
\end{equation}
where $s_{i,k}$ is the 3-dimensional position of the $i$th satellite at the time of transmission. The measurement model is linearized with respect to $x_k$ at each prior mean using the first order Taylor series approximation. The compared filters are based on three different models for the measurement noise $e_k$ where
\begin{align}
[e_k]_i &\sim \mathrm{ST}(\mu,\sigma^2,\delta,\nu) \label{eq:gnss_skewt}; \\
[e_k]_i &\sim \mathcal{T}(\mut,\st^2,\nut) \label{eq:gnss_t}; \\
[e_k]_i &\sim \N(\mun,\sn^2) \label{eq:gnss_normal} .
\end{align}
The skew-$t$ model \eqref{eq:gnss_skewt} is the basis for STF and STVBF, the $t$ model \eqref{eq:gnss_t} is the basis for TVBF, and the normal model \eqref{eq:gnss_normal} is the basis for the extended KF (EKF) with 99\,\% measurement validation gating. The pseudoranges are unbiased in the line-of-sight case, so the location parameters are fixed to $\mun\!=\!\mut\!=\!0$. Furthermore, the degrees of freedom are fixed to $\nu\!=\!\nut\!=\!4$, which according to our experience is in general a good compromise between outlier robustness and performance based on inlier measurements, provides infinite kurtosis but finite skewness and variance, and is recommended in \cite{lange1989}. The deviation parameter $\sn$ of the normal model was then fitted to the data using the expectation--maximization algorithm \cite[Ch. 12.3.3]{Sarkka10} and the parameter $\st$ of the $t$ model as well as the parameters $\sigma$ and $\delta$ of the skew-$t$ model were fitted with the particle--Metropolis algorithm \cite[Ch.\ 12.3.4]{Sarkka10}. The location parameter $\mu$ was obtained by numerically finding the point that sets the mode of the skew-$t$ noise distribution to zero. Furthermore, we added a heuristic method for mitigating the STVBF's covariance underestimation, namely a posterior covariance scaling factor that scales each STVBF posterior covariance matrix with the number $3.25^2$. This scaling was found to provide the lowest RMSE for our data set, which ensures that we do not favor the proposed STF over STVBF. These three error distributions' parameters are given in Table \ref{table:gnss_params}, and the PDFs are plotted in Fig.~\ref{fig:sim_error_pdf}.

\begin{figure}
\centering
\includegraphics[width=0.7\columnwidth]{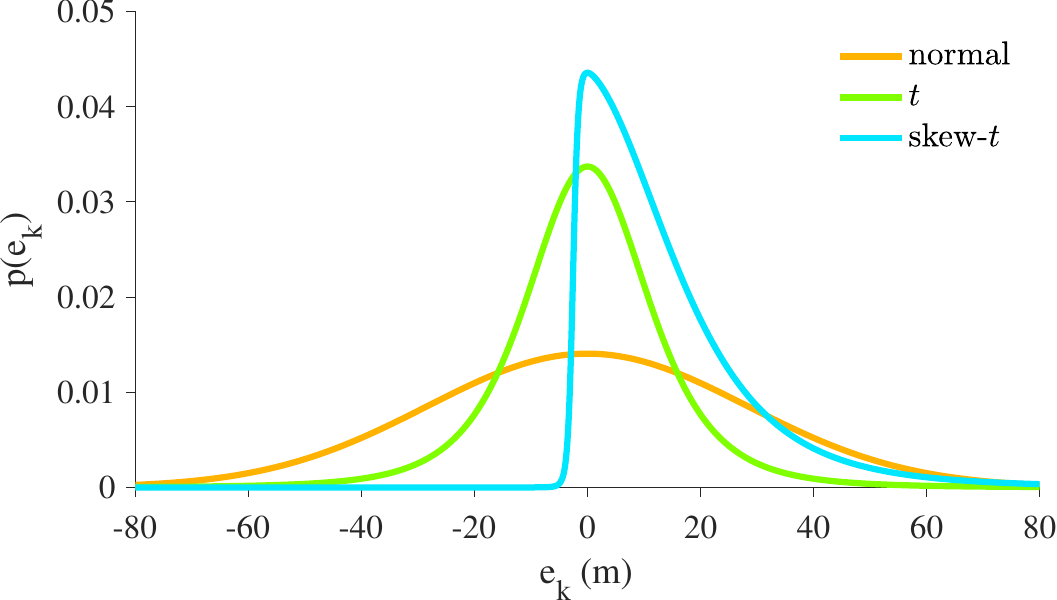}
\vspace{-3mm}
\caption{Measurement error distributions fitted to the real GNSS data for normal, $t$, and skew-$t$ error models. The modes are fixed to zero.} \label{fig:sim_error_pdf}
\end{figure}

\begin{table}[t]
\centering
\caption{Filter parameters for real GNSS data}
\label{table:gnss_params}
\begin{tabular}{ccc|cc|cc}
\multicolumn{3}{c|}{Skew-$t$, $\nu=4$} &\multicolumn{2}{c|}{$t$, $\nut=4$} & \multicolumn{2}{c}{Normal} \\
\hline\rule{0pt}{2ex}
$\mu$ (m) & $\sigma$ (m)&  $\delta$ (m) & $\mut$ (m) & $\st$ (m) & $\mun$ (m) & $\sn$ (m)  \\
\hline
\rule{0pt}{2ex}
-2.5 & 0.8 & 16.8 & 0 & 11.1 & 0 & 28.4 \\
\end{tabular}
\end{table}

Fig.\ \ref{fig:gnss_nvb_vs_err} shows the filter RMSEs as a function of the number of VB iterations. Both \STRTVBF and TVBF converge within five VB iterations, while the STVBF does not converge within 30 iterations but requires about 150 iterations. The empirical CDF graphs of the user position errors with five VB iterations for \STRTVBF and TVBF and with 150 iterations for STVBF are shown in Fig.\ \ref{fig:gnss_cdf}, and the RMSEs as well as the relative running times are given in Table \ref{tab:gnss_rmse}. The results show that modelling the skewness improves the positioning accuracy and is important especially for the accuracy in vertical direction. This can be explained by the sensitivity of the vertical direction to large measurement errors; due to bad measurement geometry the accuracy in the vertical direction is low even with line-of-sight measurements, so correct downweighting of erroneous altitude information requires careful modelling of the noise distribution's tails. The computational burden of our \STRTVBF implementation with five VB iterations is almost four times that of TVBF, but Fig.\ \ref{fig:gnss_nvb_vs_err} shows that two \STRTVBF iterations would already be enough to match TVBF's average RMSE.

Fig.\ \ref{fig:gnss_cdf} and Table \ref{tab:gnss_rmse} also show that the proposed \STRTVBF is more accurate than STVBF despite STVBF being considerable heavier computationally due to STVBF's 150 VB iterations. Furthermore, achieving this STVBF performance required awkward and data-dependent tuning to reduce the underestimation of posterior covariance matrix. The issues shown by STVBF are probably due to the highly skewed measurement noise distribution.
\begin{figure}
\centering
\includegraphics[width=0.48\columnwidth]{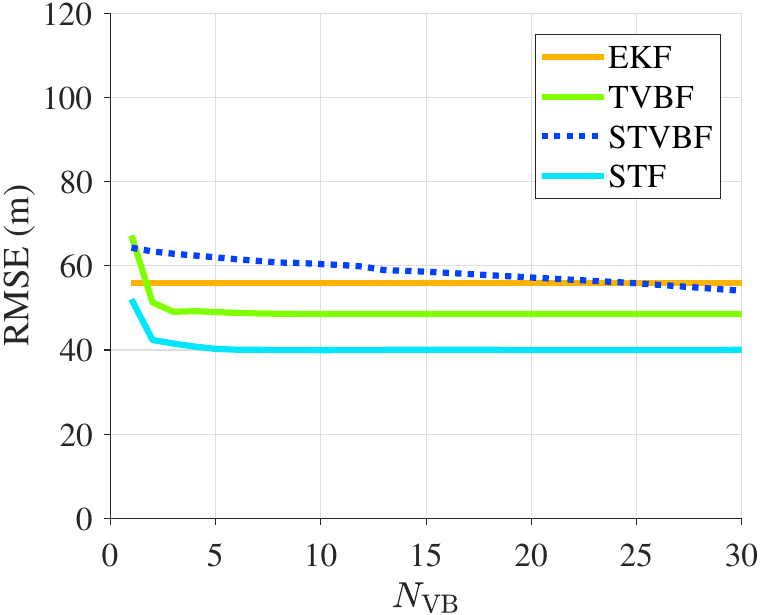}
\includegraphics[width=0.48\columnwidth]{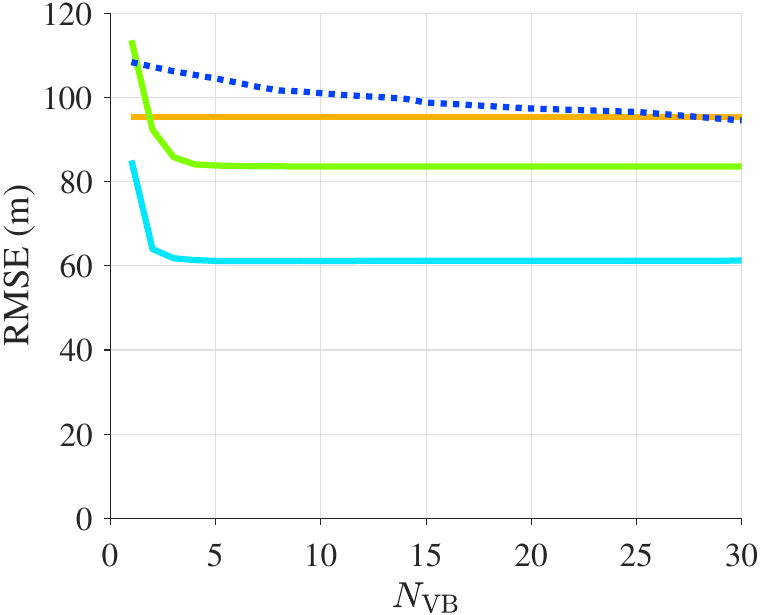}
\vspace{-3mm}
\caption{RMSE of horizontal (left) and vertical (right) position for real GNSS data as a function of the number of VB iterations} \label{fig:gnss_nvb_vs_err}
\end{figure}
\begin{figure}
\centering
\includegraphics[width=0.492\columnwidth,trim=0mm 0mm 1mm 0mm,clip=true]{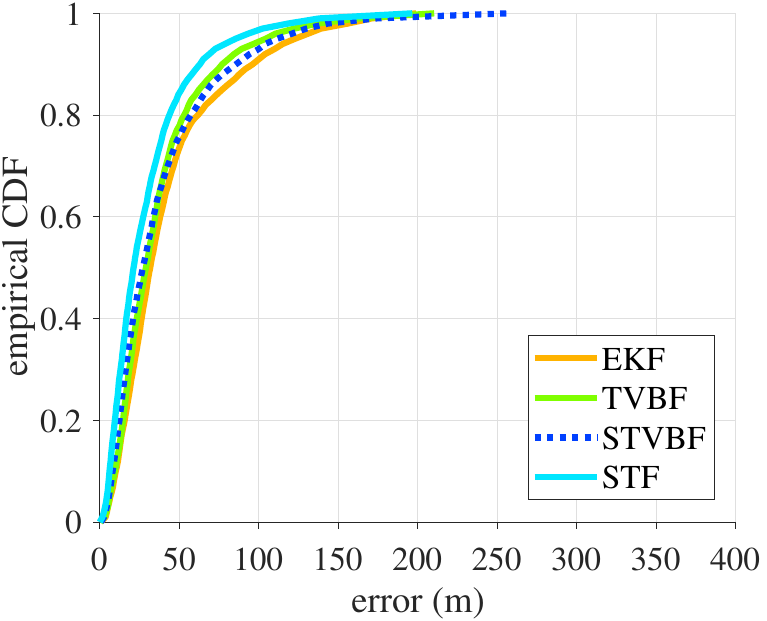}
\hfill
\includegraphics[width=0.492\columnwidth,trim=0mm 0mm 1mm 0mm,clip=true]{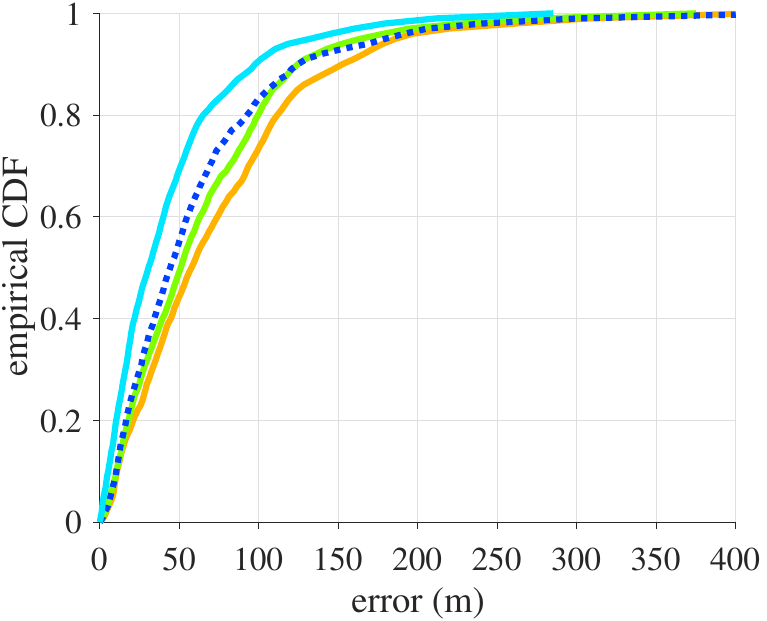}
\vspace{-6mm}
\caption{Empirical error CDFs for the real GNSS data for the horizontal error (left) and the vertical error (right)} \label{fig:gnss_cdf}
\end{figure}
\begin{table}[t]
\centering
\caption{The RMSEs and relative running times for real GNSS data}
\label{tab:gnss_rmse}
\begin{tabular}{l|c|c|c|c}
 & EKF & TVBF & STVBF & \STRTVBF  \\
\hline
\rule{-0.9ex}{2ex}
RMSE$_\text{horizontal}$ (m) & 56 & 49 & 52 & 40 \\
RMSE$_\text{vertical}$ (m) & 95 & 84 & 84 & 61 \\
\hline
\rule{-0.9ex}{2ex}
Running time & 1 & 1.3 & 18.7 & 3.8
\end{tabular}
\end{table}


\section{Extension to MVST} \label{sec:extension}
The skew $t$-distribution has several multivariate versions. In \cite{branco2001, azzalini2003, gupta2003skew} the PDF of the multivariate skew $t$-distribution (MVST) involves the CDF of a univariate $t$-distribution, while the definition of skew $t$-distribution given in \cite{sahu2003} involves the CDF of a multivariate $t$-distribution. These versions of MVST are special cases of more general multivariate skew-$t$-type distribution families, which include the multivariate canonical fundamental skew $t$-distribution (CFUST) \cite{arellanovalle2005} and the multivariate unified skew $t$-distribution \cite{arellanovalle2010}. A comprehensive review on the different variants of the MVST is given in \cite{lee2016}.

The MVST variant used in this article is based on the CFUST discussed in \cite{lee2016}, and it is the most general variant of the MVST. In this variant the parameter matrix $R\in \mathbb{R}^{n_z \times n_z}$ is a square positive-definite matrix, and $\Delta\in \mathbb{R}^{n_z \times n_z}$ is an arbitrary matrix. The PDF is
\begin{align}
\mathrm{MVST}(z;\mu,R,\Delta,\nu) 
= 2^{n_z} \mathrm{t}(z;\mu,\Omega,\nu) \, \mathrm{T}(\overline{z}; 0,L,\nu+n_z) ,
\end{align}
where $L\!=\!\eye_{n_z}-\Delta^\t \Omega^{-1} \Delta$, $\Omega=R+\Delta\Delta^\t$,
\begin{align}
&\mathrm{t}(z;\mu,\Sigma,\nu) \nonumber\\
&= \tfrac{\Gamma\left(\tfrac{\nu+n_z}{2}\right)}{(\nu\pi)^\frac{n_z}{2} \mathrm{det}(\Sigma)^\frac{1}{2}\Gamma\left(\tfrac{\nu}{2}\right)} \left( 1+ \tfrac{1}{\nu}(z-\mu)^\t\Sigma^{-1}(z-\mu) \right)^{-\frac{\nu+n_z}{2}}
\end{align}
is the PDF of the $n_z$-variate $t$-distribution and $\mathrm{T}(z;\mu,\Sigma,\nu)$ its CDF, and
\begin{equation}
\overline{z} = \Delta^\t \Omega^{-1}(z-\mu) \sqrt{\tfrac{\nu+n_z}{\nu+(z-\mu)^\t \Omega^{-1} (z-\mu)}} .
\end{equation}

The inference algorithms proposed in this paper can be extended to cover the case where the elements of the measurement noise vector are not statistically independent but jointly multivariate skew-$t$-distributed. When the measurement noise follows a MVST, i.e.
\begin{align} \label{eq:measmod_mvskewt}
e_k &\sim \mathrm{MVST}(0,R,\Delta,\nu),
\end{align}  
the smoothing and filtering algorithms presented in Algorithms \ref{table:smoothing} and \ref{table:filtering} apply with slight modifications. At the core of this convenient extension is the fact that the MVST can be represented by a similar hierarchical model as in \eqref{eq:hierarchical}. However, the shape matrices $\Delta$ and $R$ are not required to be diagonal, and the matrix $\Lambda_k$ has the form $\lambda_k \cdot I_{n_y}$, where $\lambda_k$ is a scalar with the prior
\begin{align} \label{eq:lambda_mvst}
\lambda_k &\sim \G(\tfrac{\nu}{2},\tfrac{\nu}{2}) .
\end{align}
Notice that when $\lambda_k$ admits a small value, all the measurement components can potentially be outliers simultaneously unlike with the independent univariate skew-$t$ components model. A univariate skew-$t$ is also a MVST, but a vector of univariate independently skew-$t$ distributed components is not a special case of MVST. This difference is illustrated by the PDF contour plots in Fig.~\ref{fig:2d_pdf}. See also further discussion in \cite{lee2016}.
\newsavebox{\smlmat}
\savebox{\smlmat}{$\left[\begin{smallmatrix}4\\4\end{smallmatrix}\right]$}
\begin{figure}[t]
\centering
\subfloat[]{\includegraphics[width=0.4\columnwidth]{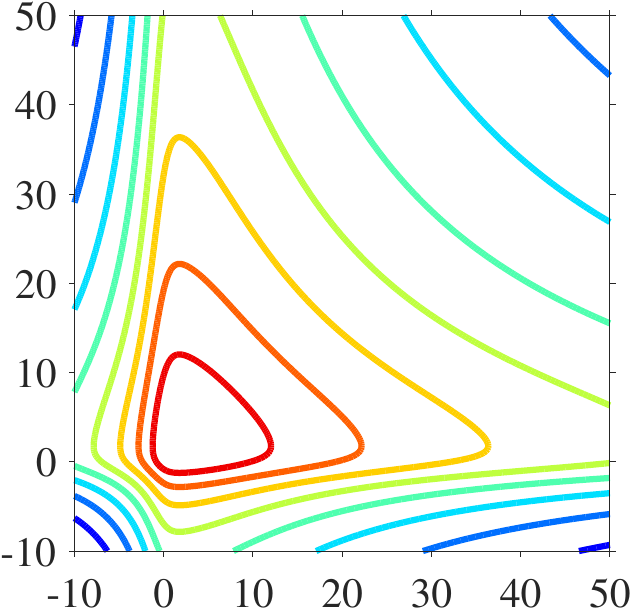}}
\quad
\subfloat[]{\includegraphics[width=0.4\columnwidth]{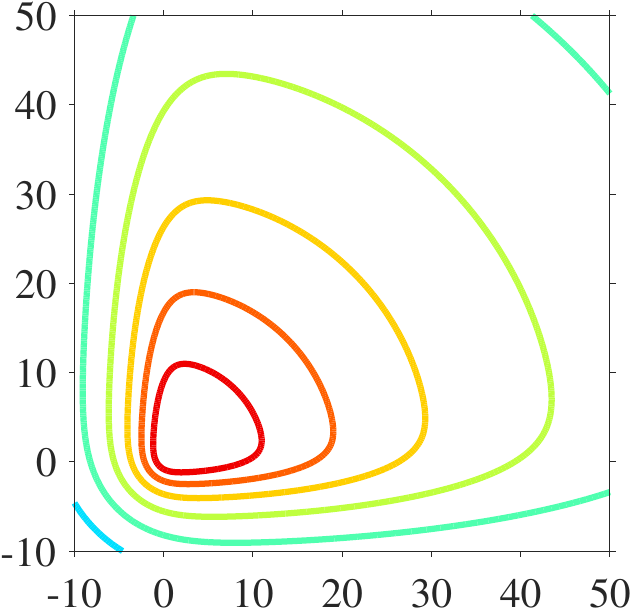}}
\vspace{-1mm}
\caption{PDF of bivariate measurement noise from (a) independent univariate skew-$t$ components model \eqref{eq:measmod_independent} with $\Delta\!=\!5 \eye_2$, $R\!=\!\eye_2$, $\nu\!=\!\usebox{\smlmat}$ and (b) MVST model \eqref{eq:measmod_mvskewt} with $\Delta\!=\!5 \eye_2$, $R\!=\!\eye_2$, $\nu\!=\!4$.}
\label{fig:2d_pdf}
\end{figure}

The specific modification required by MVST measurement noise to the STS algorithm in Algorithm~\ref{table:smoothing} is replacing line \ref{listing:lambda_smoothing} by
\begin{align}
\newcommand{\mlambda}[1]{\Lambda_{#1|K}}
\mlambda{k} \gets \frac{\nu+2n_y}{\nu+\tr\{\Psi_k\}} \cdot I_{n_y}
\end{align}
Similarly, the specific modification required by MVST measurement noise to the STF algorithm in Algorithm~\ref{table:filtering} is replacing line \ref{listing:lambda_filtering} by
\begin{align}
\Lambda_{k|k} \gets \frac{\nu+2n_y}{\nu+\tr\{\Psi_k\}} \cdot I_{n_y} .
\end{align}

\section{Performance Bound} \label{sec:bound}
\newcommand{\est}{\hat{x}}
\newcommand{\Ur}{r}

\subsection{Cram\'er--Rao lower bound}

The Bayesian Cram\'er--Rao lower bound (CRLB) $B$ is a lower bound for the mean-square-error (MSE) matrix of the state estimator $\est$ of the random variable $x$ using the observations $y$
\begin{equation}
M = \E_{p(x,y)} [(x-\est)(x-\est)^\t]
\end{equation}
in the sense that the matrix difference $M\!-\!B$ is positive semidefinite for any state estimator \cite[Ch.\ 2.4]{vantrees}. The regularity conditions sufficient for the positive-semidefiniteness to hold \cite[Ch.\ 2.4]{vantrees} are the integrability of the first two partial derivatives of the joint density $p(x_{1:k},y_{1:k})$ for an asymptotically unbiased estimator. These conditions are satisfied by the skew-$t$ likelihood and the normal prior distribution, even though they do not hold for $p(x_{1:k},u_{1:k},\Lambda_{1:k},y_{1:k})$ of the hierarchical model used in the proposed variational estimator due to restriction of $u_{1:k}$ to the positive orthant. This is sufficient, since we only seek the CRLB for the actual state $x$, not for the artificial variables $u$ and $\Lambda$.

The filtering CRLB $B_{k|k}$ for the state-space model \eqref{eq:state-evo}--\eqref{eq:measurementmodel} follows the recursion \cite{simandl2001}
\begin{subequations} \label{eq:crlb_skewtstatespace}
\begin{align}
&B_{1|0} = P_{1|0} \\
&B_{k+1|k+1} = \big((AB_{k|k}A^\t+Q)^{-1} + \E_{p(x_k|y_{1:k-1})}[\mathcal{I}(x_k)] \big)^{-1} ,
\end{align}
\end{subequations}
where $\mathcal{I}(e_k)$ is the Fisher information matrix of the measurement noise distribution. Furthermore, the smoothing CRLB for the state-space model \eqref{eq:state-evo}--\eqref{eq:measurementmodel} follows the recursion \cite{simandl2001}
\begin{align}
B_{k|K}
&= B_{k|k} + G_k (B_{k+1|K}-B_{k+1|k}) G_k^\t,
\end{align}
where
\begin{align}
&G_k = B_{k|k} A^\t B_{k+1|k}^{-1} ,\\
&B_{k+1|k} = A B_{k|k} A^\t + Q .
\end{align}
This coincides with the covariance matrix update of Rauch--Tung--Striebel smoother's backward recursion \cite{RTS-1965}.

The Fisher information matrix for the multivariate skew-$t$-distributed measurement noise of \eqref{eq:measmod_mvskewt} is
 \begin{equation} \label{eq:fisher_mvst}
     \mathcal{I}(x) = C^\t (R+\Delta\Delta^\t)^{-\frac{\t}{2}} E (R+\Delta\Delta^\t)^{-\frac{1}{2}} C,
 \end{equation}
 where
 \begin{equation} \label{eq:fisher_mvst_E}
 E = \E_{p(\Ur)} \left[ \tfrac{\nu+n_y}{\nu+\Ur^\t\Ur} \left( I_{n_y} - \tfrac{2}{\nu+\Ur^\t\Ur} \Ur\Ur^\t + \widetilde{R}_\Ur \widetilde{R}_\Ur^\t \right) \right]
 \end{equation}
with $\Ur\!\sim\!\mathrm{MVST}(0,\eye_{n_y}-\Theta\Theta^\t,\Theta,\nu)$, $\Theta=(R+\Delta\Delta^\t)^{-\frac{1}{2}}\Delta$, $A^\frac{1}{2}$ is a square-root matrix such that $A^\frac{1}{2}(A^\frac{1}{2})^\t=A$, $A^{-\frac{1}{2}}\triangleq(A^\frac{1}{2})^{-1}$, $A^{-\frac{\t}{2}}\triangleq((A^\frac{1}{2})^{-1})^\t$, and
\begin{align} \label{eq:R-r-MVST}
\widetilde{R}_\Ur &= \big(\mathrm{T}\big( \Theta^\t \Ur \sqrt{\tfrac{\nu+n_y}{\nu+\Ur^\t \Ur} }; 0,\eye_{n_y}-\Theta^\t \Theta,\nu+n_y\big) \big)^{-1} \nonumber\\
&\times (\eye_{n_y} - \tfrac{1}{\nu+\Ur^\t \Ur} \Ur\Ur^\t ) \Theta\nonumber\\
&\times \nabla_u \mathrm{T}(u;0,\eye_{n_y}-\Theta^\t \Theta,\nu+n_y)\Bigr|_{u=\Theta^\t \Ur \sqrt{\frac{\nu+n_y}{\nu+\Ur^\t \Ur}}} ,
\end{align}
where $\nabla_u$ is the gradient with respect to $u$. The derivation is given in Appendix~\ref{sec:crlb_derivations}. The evaluation of the expectation in \eqref{eq:fisher_mvst_E} is challenging with high-dimensional measurements due to the requirement to evaluate the CDF of the multivariate $t$-distribution and its partial derivatives. By the Woodbury matrix identity, the recursion \eqref{eq:crlb_skewtstatespace} is equivalent to the covariance matrix update of the Kalman filter with the measurement noise covariance $(R+\Delta\Delta^\t)^\frac{1}{2} E^{-1} ((R+\Delta\Delta^\t)^\frac{1}{2})^\t$.

In the model \eqref{eq:measmod_independent} the measurement noise components are independently univariate skew-$t$-distributed. In this case the Fisher information is obtained by applying \eqref{eq:fisher_mvst} to each conditionally independent measurement component and summing. The resulting formula matches with \eqref{eq:fisher_mvst}, the matrix $E$ now being a diagonal matrix with the diagonal entries
\begin{align} \label{eq:crlb_univ_e}
E_{ii} =& \E_{p(\Ur_i)}\! \bigg[ \tfrac{\nu_i-\Ur_i^2}{(\nu_i+\Ur_i^2)^2} \nonumber\\
+& \tfrac{\theta_i^2}{1-\theta_i^2} \tfrac{\nu_i^2}{(\nu_i+\Ur_i^2)^3}\! \left(\tau_{\nu_i+1}\big(\tfrac{\theta_i}{\sqrt{1-\theta_i^2}} \Ur_i \sqrt{\tfrac{\nu_i+1}{\nu_i+\Ur_i^2}}\big)\right)^2 \bigg] ,
\end{align}
where $\Ur_i \!\sim\! \mathrm{ST}(0,1\!-\!\theta_i^2,\theta_i,\nu_i)$ is a univariate skew-$t$-distributed random variable, $\theta_i\!=\!\Delta_{ii}/\sqrt{R_{ii}+\Delta_{ii}^2}$ and $\tau_\nu(x) \!=\! \mathrm{t}(x;0,1,\nu)/\mathrm{T}(x;0,1,\nu)$. Substituted into \eqref{eq:fisher_mvst}, this formula matches the Fisher information formula obtained for the univariate skew $t$-distribution in \cite{diciccio2011}. In this case only integrals with respect to one scalar variable are to be evaluated numerically.

\subsection{Simulation}

We study the CRLB in \eqref{eq:crlb_skewtstatespace} of a linear state-space model with skew-$t$-distributed measurement noise by generating realizations of the model
\begin{subequations} \label{eq:crlbtest_model}
\begin{align}
x_{k+1} &= \left[ \begin{smallmatrix} 1&1\\0&1 \end{smallmatrix} \right] x_k + w_k,\ w_k \sim \N(0,Q) \\
y_k &= \left[ \begin{smallmatrix} 1&0 \end{smallmatrix} \right] x_k + e_k,\ e_k \sim \mathrm{ST}(\mu,\sigma^2,\delta,\nu) ,
\end{align}
\end{subequations}
where $x\in\mathbb{R}^2$ is the state, $Q=\left[ \begin{smallmatrix} 1/3 & 1/2\\1/2&1 \end{smallmatrix} \right]$ is the process noise covariance matrix, $y_k\in\mathbb{R}$ is the measurement, and $\nu$ and $\delta_c$ are parameters that determine other parameters by the formulas
\begin{subequations}
\begin{align}
\mu &= -\gamma \delta_c \sigma, \\
\sigma^2 & = \tfrac{\omega^2}{\frac{\nu}{\nu-2}(1+\delta_c^2)-\gamma^2 \delta_c^2} ,\\
\delta &= \delta_c \sigma , \\
\gamma&=\sqrt{\tfrac{\nu}{\pi}} \tfrac{\Gamma((\nu-1)/2)}{\Gamma(\nu/2)} .
\end{align}
\end{subequations}
Thus, the measurement noise distribution is zero-mean and has the variance $\omega^2\!=\!5^2$. We generate 10\,000 realizations of a 50-step process, and compute the CRLB and mean-square-errors (MSE) of the bootstrap PF with 2000 particles and the \STRTVBF. The CRLB and the MSEs were computed for the first component of the state at the last time instant $[x_{50}]_1$.

Fig.\ \ref{fig:crlb} shows the CRLB of the model \eqref{eq:crlbtest_model}. The figure shows that increase in the skewness as well as heavy-tailedness can decrease the CRLB significantly, which suggests that a nonlinear filter can be significantly better than the KF, which gives MSE 11.8 for all $\delta_c$ and $\nu$. Fig.\ \ref{fig:crlb_comp} shows the MSEs of PF and \STRTVBF. As expected, when $\nu\rightarrow\infty$ and $\delta_c\rightarrow0$, the PF's MSE approaches the CRLB. \STRTVBF is only slightly worse than PF. The figures also show that although the CRLB becomes looser when the distribution becomes more skewed and/or heavy-tailed, it correctly indicates that modeling the skewness still improves the filtering performance.
\begin{figure}[htbp]
   \centering
   \includegraphics[width=0.75\columnwidth]{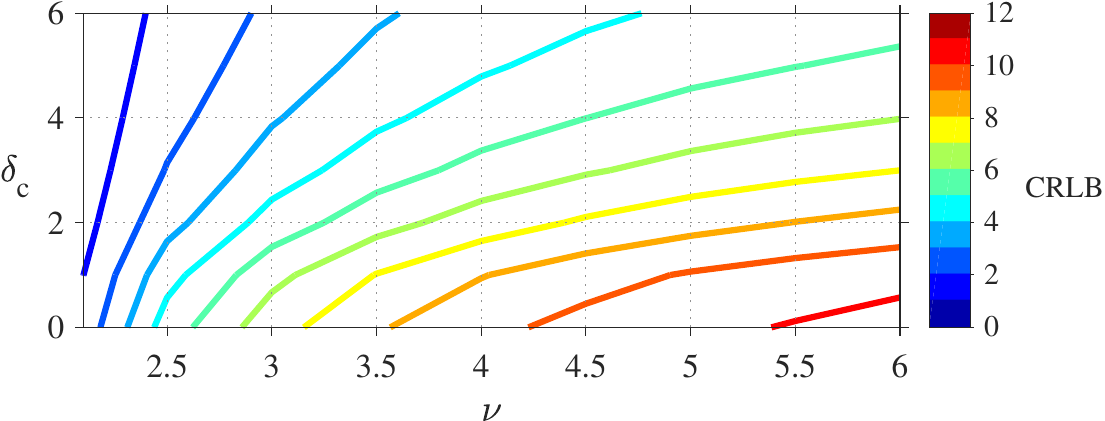}
   \vspace{-3mm}
   \caption{The CRLB of the 50th time instant for the model \eqref{eq:crlbtest_model} with a fixed measurement noise variance. Skewness and heavy-tailedness decreases the CRLB significantly.}
   \label{fig:crlb}
\end{figure}
\vspace{-4mm}
\begin{figure}[htbp]
   \centering
   \includegraphics[width=0.445\columnwidth]{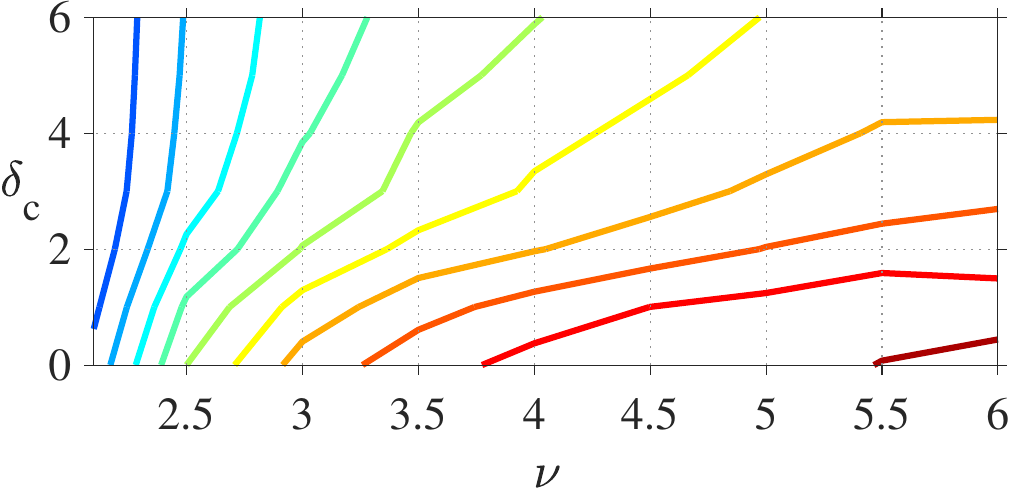}
   \hfill
   \includegraphics[width=0.445\columnwidth]{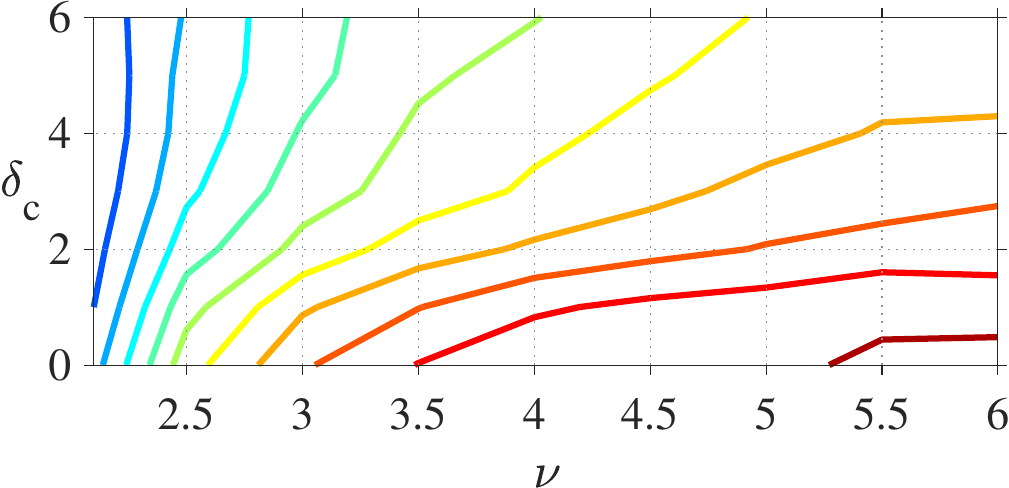}
   \includegraphics[width=0.078\columnwidth,trim=147mm -9mm 0mm 0mm,clip]{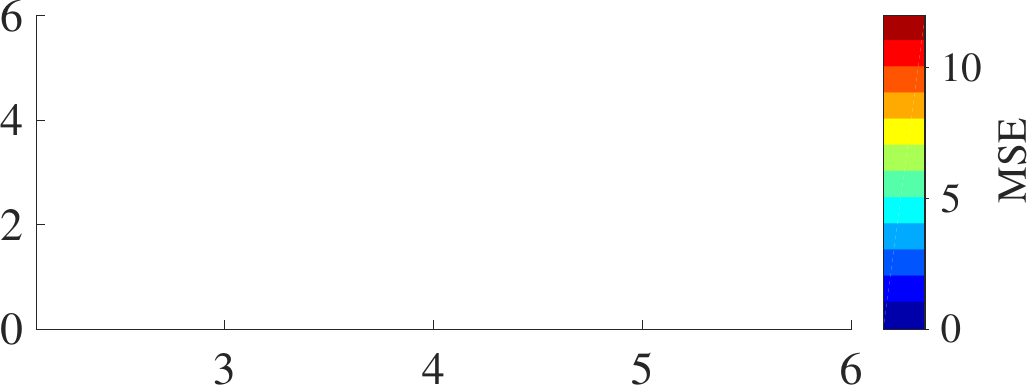}
   \vspace{-5mm}
   \caption{The MSEs of PF (left) and \STRTVBF's (right) are close to each other.}
   \label{fig:crlb_comp}
\end{figure}
\section{Conclusions} \label{sec:conclusions}
We have proposed a novel approximate filter and smoother for linear state-space models with heavy-tailed and skewed measurement noise distribution, and derived the Cram\'er--Rao lower bounds for the filtering and smoothing estimators. The algorithms are based on the variational Bayes approximation, where some posterior independence approximations are removed from the earlier versions of the algorithms to avoid significant underestimation of the posterior covariance matrix. Removal of independence approximations is enabled by the expectation propagation (EP) algorithm for approximating the mean and covariance matrix of truncated multivariate normal distribution. A greedy processing sequence is given for the EP. Simulations and real-data tests with GNSS positioning data show that the proposed algorithms outperform the state-of-the-art low-complexity methods, including the earlier skew-$t$ VB filter, in a real-world estimation problem.
\bibliographystyle{IEEEtran}
\bibliography{IEEEabrv,VB-Skewness}

\begin{thebibliography}{10}
\providecommand{\url}[1]{#1}
\csname url@samestyle\endcsname
\providecommand{\newblock}{\relax}
\providecommand{\bibinfo}[2]{#2}
\providecommand{\BIBentrySTDinterwordspacing}{\spaceskip=0pt\relax}
\providecommand{\BIBentryALTinterwordstretchfactor}{4}
\providecommand{\BIBentryALTinterwordspacing}{\spaceskip=\fontdimen2\font plus
\BIBentryALTinterwordstretchfactor\fontdimen3\font minus
  \fontdimen4\font\relax}
\providecommand{\BIBforeignlanguage}[2]{{%
\expandafter\ifx\csname l@#1\endcsname\relax
\typeout{** WARNING: IEEEtran.bst: No hyphenation pattern has been}%
\typeout{** loaded for the language `#1'. Using the pattern for}%
\typeout{** the default language instead.}%
\else
\language=\csname l@#1\endcsname
\fi
#2}}
\providecommand{\BIBdecl}{\relax}
\BIBdecl

\bibitem{GusGun2005}
F.~Gustafsson and F.~Gunnarsson, ``Mobile positioning using wireless networks:
  possibilities and fundamental limitations based on available wireless network
  measurements,'' \emph{IEEE Signal Processing Magazine}, vol.~22, no.~4, pp.
  41--53, July 2005.

\bibitem{BorsenChen2009}
B.-S. Chen, C.-Y. Yang, F.-K. Liao, and J.-F. Liao, ``Mobile location estimator
  in a rough wireless environment using {E}xtended {K}alman-based {IMM} and
  data fusion,'' \emph{IEEE Transactions on Vehicular Technology}, vol.~58,
  no.~3, pp. 1157--1169, March 2009.

\bibitem{Kok2015}
M.~Kok, J.~D. Hol, and T.~B. Sch\"{o}n, ``Indoor positioning using
  ultra-wideband and inertial measurements,'' \emph{IEEE Transactions on
  Vehicular Technology}, vol.~64, no.~4, 2015.

\bibitem{kaemarungsi2012}
K.~Kaemarungsi and P.~Krishnamurthy, ``Analysis of {WLAN}'s received signal
  strength indication for indoor location fingerprinting,'' \emph{Pervasive and
  Mobile Computing}, vol.~8, no.~2, pp. 292--316, 2012, special Issue:
  Wide-Scale Vehicular Sensor Networks and Mobile Sensing.

\bibitem{branco2001}
M.~D. Branco and D.~K. Dey, ``A general class of multivariate skew-elliptical
  distributions,'' \emph{Journal of Multivariate Analysis}, vol.~79, no.~1, pp.
  99--113, October 2001.

\bibitem{azzalini2003}
A.~Azzalini and A.~Capitanio, ``\BIBforeignlanguage{English}{Distributions
  generated by perturbation of symmetry with emphasis on a multivariate skew
  $t$-distribution},'' \emph{\BIBforeignlanguage{English}{Journal of the Royal
  Statistical Society. Series B (Statistical Methodology)}}, vol.~65, no.~2,
  pp. 367--389, 2003.

\bibitem{gupta2003skew}
A.~K. Gupta, ``Multivariate skew $t$-distribution,'' \emph{Statistics},
  vol.~37, no.~4, pp. 359--363, 2003.

\bibitem{nurminen2015b}
H.~Nurminen, T.~Ardeshiri, R.~Pich\'{e}, and F.~Gustafsson, ``A {NLOS}-robust
  {TOA} positioning filter based on a skew-$t$ measurement noise model,'' in
  \emph{International Conference on Indoor Positioning and Indoor Navigation
  (IPIN)}, October 2015, pp. 1--7.

\bibitem{fruhwirth2010}
S.~Fr\"uhwirth-Schnatter and S.~Pyne, ``Bayesian inference for finite mixtures
  of univariate and multivariate skew-normal and skew-$t$ distributions,''
  \emph{Biostatistics}, vol.~11, no.~2, pp. 317--336, 2010.

\bibitem{counsell2010}
N.~Counsell, M.~Cortina-Borja, A.~Lehtonen, and A.~Stein, ``Modelling
  psychiatric measures using skew-normal distributions,'' \emph{European
  Psychiatry}, vol.~26, no.~2, pp. 112--114, 2010.

\bibitem{eling2012}
M.~Eling, ``Fitting insurance claims to skewed distributions: Are the
  skew-normal and skew-student good models?'' \emph{Insurance: Mathematics and
  Economics}, vol.~51, no.~2, pp. 239--248, 2012.

\bibitem{marchenko2010phd}
Y.~V. Marchenko, ``Multivariate skew-$t$ distributions in econometrics and
  environmetrics,'' Ph.D. dissertation, Texas A\&M University, December 2010.

\bibitem{doucet2000}
A.~Doucet, S.~Godsill, and C.~Andrieu, ``On sequential {M}onte {C}arlo sampling
  methods for {B}ayesian filtering,'' \emph{Statistics and Computing}, vol.~10,
  no.~3, pp. 197--208, July 2000.

\bibitem{naveau2005}
P.~Naveau, M.~G. Genton, and X.~Shen, ``A skewed {K}alman filter,''
  \emph{Journal of Multivariate Analysis}, vol.~94, pp. 382--400, 2005.

\bibitem{kim2014}
H.-M. Kim, D.~Ryu, B.~K. Mallick, and M.~G. Genton, ``Mixtures of skewed
  {K}alman filters,'' \emph{Journal of Multivariate Analysis}, vol. 123, pp.
  228--251, 2014.

\bibitem{Rezaie2014}
J.~Rezaie and J.~Eidsvik, ``Kalman filter variants in the closed skew normal
  setting,'' \emph{Computational Statistics and Data Analysis}, vol.~75, pp.
  1--14, 2014.

\bibitem{gsum1972}
D.~L. Alspach and H.~W. Sorenson, ``Nonlinear {B}ayesian estimation using
  {G}aussian sum approximations,'' \emph{IEEE Transactions on Automatic
  Control}, vol.~17, no.~4, pp. 439--448, Aug. 1972.

\bibitem{bar1988tracking}
Y.~Bar-Shalom and T.~Fortmann, \emph{Tracking and Data Association}, ser.
  Mathematics in Science and Engineering Series.\hskip 1em plus 0.5em minus
  0.4em\relax Academic Press, 1988.

\bibitem{maybeck2006}
J.~L. Williams and P.~S. Maybeck, ``{Cost-function-based hypothesis control
  techniques for multiple hypothesis tracking},'' \emph{Mathematical and
  Computer Modelling}, vol.~43, no. 9--10, pp. 976--989, May 2006.

\bibitem{nurminen2015a}
H.~Nurminen, T.~Ardeshiri, R.~Piche, and F.~Gustafsson, ``Robust inference for
  state-space models with skewed measurement noise,'' \emph{IEEE Signal
  Processing Letters}, vol.~22, no.~11, pp. 1898--1902, 2015.

\bibitem{wand2011}
M.~P. Wand, J.~T. Ormerod, S.~A. Padoan, and R.~Fr\"uhwirth, ``Mean field
  variational {B}ayes for elaborate distributions,'' \emph{Bayesian Analysis},
  vol.~6, no.~4, pp. 847--900, 2011.

\bibitem{Bishop2007}
C.~M. Bishop, \emph{Pattern Recognition and Machine Learning}.\hskip 1em plus
  0.5em minus 0.4em\relax Springer, 2007.

\bibitem{lee2016}
S.~X. Lee and G.~J. McLachlan, ``Finite mixtures of canonical fundamental skew
  $t$-distributions -- the unification of the restricted and unrestricted skew
  $t$-mixture models,'' \emph{Statistics and {C}omputing}, no.~26, pp.
  573--589, 2016.

\bibitem{CoverT2006}
T.~M. Cover and J.~Thomas, \emph{Elements of Information Theory}.\hskip 1em
  plus 0.5em minus 0.4em\relax John Wiley and Sons, 2006.

\bibitem{TzikasLG2008}
D.~G. Tzikas, A.~C. Likas, and N.~P. Galatsanos, ``The variational
  approximation for {B}ayesian inference,'' \emph{IEEE Signal Processing
  Magazine}, vol.~25, no.~6, pp. 131--146, Nov. 2008.

\bibitem{Beal03}
M.~J. Beal, ``Variational algorithms for approximate {B}ayesian inference,''
  Ph.D. dissertation, Gatsby Computational Neuroscience Unit, University
  College London, 2003.

\bibitem{RTS-1965}
H.~E. Rauch, C.~T. Striebel, and F.~Tung, ``{Maximum Likelihood Estimates of
  Linear Dynamic Systems},'' \emph{Journal of the American Institute of
  Aeronautics and Astronautics}, vol.~3, no.~8, pp. 1445--1450, 1965.

\bibitem{tallis1961}
G.~Tallis, ``The moment generating function of the truncated multi-normal
  distribution,'' \emph{Journal of the Royal Statistical Society. Series B
  (Methodological)}, vol.~23, no.~1, pp. 223--119, 1961.

\bibitem{genz2004}
A.~Genz, ``Numerical computation of rectangular bivariate and trivariate normal
  and $t$ probabilities,'' \emph{Statistics and {C}omputing}, vol.~14, pp.
  251--260, 2004.

\bibitem{Genz02}
A.~Genz and F.~Bretz, ``Comparison of methods for the computation of
  multivariate $t$ probabilities,'' \emph{Journal of Computational and
  Graphical Statistics}, vol.~11, no.~4, pp. 950--971, 2002.

\bibitem{minka2001}
T.~P. Minka, ``Expectation propagation for approximate {B}ayesian inference,''
  in \emph{17th Annual Conference on Uncertainty in Artificial Intelligence
  {(UAI)}}, 2001, pp. 362--369.

\bibitem{cunningham2013_arxiv}
\BIBentryALTinterwordspacing
J.~P. Cunningham, P.~Hennig, and S.~Lacoste-Julien, ``{G}aussian probabilities
  and expectation propagation,'' Arxiv, November 2013. [Online]. Available:
  \url{arxiv.org/abs/1111.6832}
\BIBentrySTDinterwordspacing

\bibitem{murphy2012}
K.~P. Murphy, \emph{Machine Learning: A Probabilistic Perspective}.\hskip 1em
  plus 0.5em minus 0.4em\relax Cambridge, MA: The MIT Press, 2012.

\bibitem{dow2009}
J.~M. Dow, R.~Neilan, and C.~Rizos, ``The international {GNSS} service in a
  changing landscape of global navigation satellite systems,'' \emph{Journal of
  Geodesy}, vol.~83, no.~7, p. 689, February 2009.

\bibitem{bar-shalom}
Y.~Bar-Shalom, R.~X. Li, and T.~Kirubarajan, \emph{Estimation with Applications
  to Tracking and Navigation, Theory Algorithms and Software}.\hskip 1em plus
  0.5em minus 0.4em\relax John Wiley \& Sons, 2001.

\bibitem{piche2012}
R.~Pich\'e, S.~S\"{a}rkk\"{a}, and J.~Hartikainen, ``Recursive outlier-robust
  filtering and smoothing for nonlinear systems using the multivariate
  {S}tudent-$t$ distribution,'' in \emph{IEEE International Workshop on Machine
  Learning for Signal Processing (MLSP)}, September 2012.

\bibitem{axelrad1994}
P.~Axelrad and R.~Brown, ``{GPS} navigation algorithms,'' in \emph{{G}lobal
  {P}ositioning {S}ystem: Theory and Applications I}, B.~W. Parkinson and J.~J.
  Spilker~Jr., Eds.\hskip 1em plus 0.5em minus 0.4em\relax Washington D.C.:
  AIAA, 1996, ch.~9.

\bibitem{lange1989}
K.~L. Lange, R.~J. Little, and J.~M. Taylor, ``Robust statistical modeling
  using the $t$ distribution,'' \emph{Journal of the Americal Statistical
  Association}, vol.~84, no. 408, pp. 881--896, December 1989.

\bibitem{Sarkka10}
S.~S\"arkk\"a and J.~Hartikainen, ``On {G}aussian optimal smoothing of
  non-linear state space models,'' \emph{IEEE Transactions on Automatic
  Control}, vol.~55, no.~8, pp. 1938--1941, August 2010.

\bibitem{sahu2003}
S.~K. Sahu, D.~K. Dey, and M.~D. Branco, ``A new class of multivariate skew
  distributions with applications to {B}ayesian regression models,''
  \emph{Canadian Journal of Statistics}, vol.~31, no.~2, pp. 129--150, 2003.

\bibitem{arellanovalle2005}
R.~B. Arellano-Valle and M.~G. Genton, ``On fundamental skew distributions,''
  \emph{Journal of Multivariate Analysis}, no.~96, pp. 93--116, 2005.

\bibitem{arellanovalle2010}
------, ``Multivariate extended skew-$t$ distributions and related families,''
  \emph{METRON - International Journal of Statistics}, vol.~68, no.~3, pp.
  201--234, 2010.

\bibitem{vantrees}
H.~L. Van~Trees, \emph{Detection, Estimation, and Modulation Theory, Part I:
  Detection, Estimation, and Linear Modulation Theory}.\hskip 1em plus 0.5em
  minus 0.4em\relax New York: John Wiley \& Sons, Inc., 1968.

\bibitem{simandl2001}
M.~\v{S}imandl, J.~Kr\'alovec, and P.~Tichavsk\'y, ``Filtering, predictive, and
  smoothing {C}ram\'er--{R}ao bounds for discrete-time nonlinear dynamic
  systems,'' \emph{Automatica}, vol.~37, pp. 1703--1716, 2001.

\bibitem{diciccio2011}
T.~J. Di~Ciccio and A.~C. Monti, ``Inferential aspects of the skew
  $t$-distribution,'' \emph{Quaderni di Statistica}, vol.~13, pp. 1--21, 2011.

\bibitem{sarkka2013}
S.~S\"arkk\"a, \emph{Bayesian Filtering and Smoothing}.\hskip 1em plus 0.5em
  minus 0.4em\relax Cambridge, UK: Cambridge University Press, 2013.

\bibitem{piche2016b}
R.~Pich\'e, ``{C}ram\'er-{R}ao lower bound for linear filtering with
  t-distributed measurements,'' in \emph{19th International Conference on
  Information Fusion (FUSION)}, July 2016, pp. 536--540.

\end{thebibliography}

\appendices

\allowdisplaybreaks
\section{Derivations for the skew-$t$ smoother} \label{sec:smoother}

\subsection{Derivations for $q_{xu}$} \label{sec:Smootherq_xu}
Eq.\ \eqref{eqn:IterativeOptimizationxu} gives
\begin{align}
&\log q_{xu}(x_{1:K},u_{1:K}) = \log\N(x_1; x_{1|0}, P_{1|0}) \nonumber\\
&+\sum_{l=1}^{K-1} \log\N(x_{l+1}; Ax_l, Q) \nonumber\\
&+ \sum_{k=1}^K \E_{q_\Lambda} [\log\N(y_k; Cx_k \!+\! \Delta u_k, \Lambda_k^{-1} R) \nonumber\\
&+\log\N_+(u_k;0,\Lambda_k^{-1})] + c \\
=& \log\N(x_1; x_{1|0}, P_{1|0}) + \sum_{l=1}^{K-1} \log\N(x_{l+1}; Ax_l, Q) \nonumber\\
&- \frac{1}{2} \sum_{k=1}^K \E_{q_\Lambda} [(y_k\!-\!Cx_k\!-\!\Delta u_k)^\t R^{-1} \Lambda_k (y_k\!-\!Cx_k\!-\!\Delta u_k) \nonumber\\
&+ u_k^\t \Lambda_k u_k] + c \\
=& \log\N(x_1; x_{1|0}, P_{1|0}) + \sum_{l=1}^{K-1} \log\N(x_{l+1}; Ax_l, Q) \nonumber\\
&- \frac{1}{2} \sum_{k=1}^K \{ (y_k\!-\!Cx_k\!-\!\Delta u_k)^\t R^{-1} \Lambda_{k|K} (y_k\!-\!Cx_k\!-\!\Delta u_k) \nonumber\\
&+ u_k^\t \Lambda_{k|K} u_k \} + c \\
=&  \log\N(x_1; x_{1|0}, P_{1|0}) + \sum_{l=1}^{K-1} \log\N(x_{l+1}; Ax_l, Q) \nonumber\\
&+ \sum_{k=1}^K \{ \log\N(y_k; Ax_k + \Delta u_k, \Lambda_{k|K}^{-1} R) \nonumber\\
&+ \log\N(u_k;0,\Lambda_{k|K}^{-1}) \} + c \\
=&  \log\N\left( \left[\begin{smallmatrix} x_1 \\ u_1 \end{smallmatrix}\right] ; \left[\begin{smallmatrix} x_{1|0}\\0 \end{smallmatrix}\right] , \left[\begin{smallmatrix} P_{1|0} & \zeros \\\zeros & \Lambda_{1|K}^{-1} \end{smallmatrix}\right] \right) \nonumber\\
&+ \sum_{l=1}^{K-1} \log \N\left( \left[\begin{smallmatrix} x_{l+1}\\u_{l+1} \end{smallmatrix}\right]; \left[\begin{smallmatrix} A & \zeros \\\zeros & \zeros \end{smallmatrix}\right] \left[\begin{smallmatrix} x_l\\u_l \end{smallmatrix}\right], \left[\begin{smallmatrix} Q & \zeros \\\zeros & \Lambda_{l+1|K}^{-1} \end{smallmatrix}\right] \right) \nonumber\\
&+ \log\N\left( y_k; \left[\begin{smallmatrix} C & \Delta \end{smallmatrix}\right] \left[\begin{smallmatrix} x_k\\u_k \end{smallmatrix}\right], \Lambda_{k|K}^{-1} R \right) + c, u_{1:K}\!\geq\!0 ,
\end{align}
where $c$ is a term that is constant with respect to $(x_{1:K},u_{1:K})$ but admits different values in different equations, $\Lambda_{k|K} \triangleq \E_{q_\Lambda}[\Lambda_k]$ is derived in Appendix \ref{sec:smoother}, Subsection B, and $u_{1:K}\geq0$ means that all the components of all $u_k$ are required to be nonnegative for each $k=1\cdots K$. Up to the truncation of the $u$ components, $q_{xu}(x_{1:K}, u_{1:K})$ has thus the same form as the joint smoothing posterior of a linear state-space model with the state transition matrix $\widetilde{A} \triangleq \left[ \begin{smallmatrix} A&\zeros\\\zeros&\zeros \end{smallmatrix} \right]$, process noise covariance matrix $\widetilde{Q}_k \triangleq \left[ \begin{smallmatrix} Q&\zeros\\\zeros&\Lambda_{k+1|K}^{-1} \end{smallmatrix} \right]$, measurement model matrix $\widetilde{C} \triangleq \left[ \begin{smallmatrix} C & \Delta  \end{smallmatrix} \right]$, and measurement noise covariance matrix $\widetilde{R} \triangleq \Lambda_{k|K}^{-1} R$. We denote the PDFs related to this state-space model with $\widetilde{p}$.

It would be possible to compute the truncated multivariate normal posterior of the joint smoothing distribution $\widetilde{p}\left(\left[\begin{smallmatrix}x_{1:K}\\u_{1:K}\end{smallmatrix}\right]|y_{1:K}\right)$, and account for the truncation of $u_{1:K}$ to the positive orthant using the \recursive truncation. However, this would be impractical with large $K$ due to the large dimensionality $K\times (n_x+n_y)$. A feasible solution is to approximate each filtering distribution in the Rauch--Tung--Striebel smoother's (RTSS \cite{RTS-1965}) forward filtering step with a multivariate normal distribution by
\begin{align}
\widetilde{p}( x_k,u_k|y_{1:k}) &= \tfrac{1}{C}\,\N\left(\left[\begin{smallmatrix} x_k\\u_k \end{smallmatrix}\right]; z_{k|k}', Z_{k|k}' \right) \cdot \iverson{u_k\!\geq\!0}\\
&\approx \N\left(\left[\begin{smallmatrix} x_k\\u_k \end{smallmatrix}\right]; z_{k|k}, Z_{k|k} \right) 
\end{align}
for each $k=1\cdots K$, where $\iverson{u_k\!\geq\!0}$ is the Iverson bracket notation,
$C$ is the normalization factor, and $z_{k|k} \triangleq \E_{\widetilde{p}}\left[\left[\begin{smallmatrix}x_k\\u_k\end{smallmatrix}\right]|y_{1:k}\right]$ and $Z_{k|k} \triangleq \var_{\widetilde{p}}\left[\left[\begin{smallmatrix}x_k\\u_k\end{smallmatrix}\right]|y_{1:k}\right]$ are approximated using the \recursive truncation. Given the multivariate normal approximations of the filtering posteriors $\widetilde{p}(x_k,u_k|y_{1:k})$, by Lemma \ref{lem:gaussian_smoother} the backward recursion of the RTSS gives multivariate normal approximations of the smoothing posteriors $\widetilde{p}(x_k,u_k|y_{1:K})$. The quantities required in the derivations of Subsection B are the expectations of the smoother posteriors $x_{k|K} \triangleq \E_{q_{xu}}[x_k]$, $u_{k|K} \triangleq \E_{q_{xu}}[u_k]$, and the covariance matrices $Z_{k|K} \triangleq \var_{q_{xu}}\left[\begin{smallmatrix} x_k\\u_k \end{smallmatrix}\right]$ and $U_{k|K} \triangleq \var_{q_{xu}}[u_k]$.

\begin{lemma} \label{lem:gaussian_smoother}
\newcommand{\xu}{z}
\newcommand{\cxu}{Z}
Let $\{\xu_k\}_{k=1}^K$ be a linear--Gaussian process, and $\{y_k\}_{k=1}^K$ a measurement process such that
\begin{subequations}
\begin{align}
\xu_1 &\sim \N(\xu_{1|0},\cxu_{1|0}) \\
\xu_k | \xu_{k-1} &\sim \N(A \xu_{k-1}, Q) \\
y_k | \xu_k &\sim p(y_k|z_k),
\end{align}
\end{subequations}
where $p(y_k|z_k)$ is a known distribution, and the standard Markovianity assumptions hold. Then, if the filtering posterior $p(\xu_k|y_{1:k})$ is a multivariate normal distribution for each $k$, then for each $k < K$ holds $\xu_k | y_{1:K} \sim \N(\xu_{k|K}, \cxu_{k|K})$, where
\begin{align}
\xu_{k|K} &= \xu_{k|k} + G_k(\xu_{k+1|K}-A \xu_{k|k}) ,\\
\cxu_{k|K} &= \cxu_{k|k} + G_k(\cxu_{k+1|K}-A\cxu_{k|k}A^\t-Q)G_k^\t ,\\
G_k &= \cxu_{k|k} A^\t (A\cxu_{k|k}A^\t + Q)^{-1} ,
\end{align}
and $\xu_{k|k}$ and $\cxu_{k|k}$ are the mean and covariance matrix of the filtering posterior $p(\xu_k|y_{1:k})$.
\end{lemma}
\begin{proof}
The details are omitted here because the proof is mostly similar to that of \cite[Theorem 8.2]{sarkka2013}.
\end{proof}

\subsection{Derivations for $q_\Lambda$} \label{sec:Smootherq_L}

Eq.\ \eqref{eqn:IterativeOptimizationL} gives
\begin{align}
\log& q_\Lambda(\Lambda_{1:K})=\sum_{k=1}^K \big\{\E_{q_{xu}}[\log p(y_k | x_k,u_k,\Lambda_k) + \log p(u_k | \Lambda_k)] \nonumber\\
& + \log p(\Lambda_k)\big\}+c.
\end{align}
Thus, $q_\Lambda(\Lambda_{1:K})=\prod_{k=1}^K q_\Lambda(\Lambda_{k})$.

In the model with independent univariate skew-$t$-distributed measurement noise components \eqref{eq:measmod_independent}, the diagonal entries of $\Lambda_k$ are separate random variables, as given in \eqref{eq:hierarchical_lambda}. Therefore,
\begin{align}
&\log q_\Lambda(\Lambda_k) \nonumber\\
=& -\frac{1}{2}\E_{{q}_{xu}}\big[\tr\{(y_k\!-\!Cx_k\!-\!\Delta u_k)(y_k\!-\!Cx_k\!-\!\Delta u_k)^\t R^{-1}\Lambda_k\}\nonumber\\
& +\!\tr\{ u_k  u_k^\t{}\Lambda_k\}\big] \!+\! \sum_{i=1}^{n_y}\left( \tfrac{\nu_i}{2}\log[\Lambda_k]_{ii}- \tfrac{\nu_i}{2} [\Lambda_k]_{ii} \right)+c\\
=&-\frac{1}{2} \tr\big\{ \big[ \big((y_k\!-\!Cx_{k|K}\!-\!\Delta u_{k|K})(y_k\!-\!Cx_{k|K}\!-\!\Delta u_{k|K})^\t \nonumber\\
&+\left[\begin{smallmatrix}C&\Delta\end{smallmatrix}\right] Z_{k|K} \left[\begin{smallmatrix}C^\t\\\Delta^\t\end{smallmatrix}\right]\big) R^{-1} + (u_{k|K} u_{k|K}^\t\!+\!U_{k|K}) \big] \Lambda_k \big\} \nonumber\\
&+ \sum_{i=1}^{n_y}\left( \tfrac{\nu_i}{2}\log[\Lambda_k]_{ii}- \tfrac{\nu_i}{2} [\Lambda_k]_{ii} \right)+c \\
=& \sum_{i=1}^{n_y} \left( \tfrac{\nu_i}{2} \log[\Lambda_k]_{ii} - \tfrac{\nu_i+\Psi_{ii}}{2}[\Lambda_k]_{ii} \right) +c ,
\end{align}
where
\begin{align}
\Psi =& (y_k\!-\!Cx_{k|K}\!-\!\Delta u_{k|K})(y_k\!-\!Cx_{k|K}\!-\!\Delta u_{k|K})^\t R^{-1} \nonumber\\
&+ \left[\begin{smallmatrix}C&\Delta\end{smallmatrix}\right] Z_{k|K} \left[\begin{smallmatrix} C^\t\\\Delta^\t\end{smallmatrix}\right] R^{-1} +u_{k|K}u_{k|K}^\t+U_{k|K}. \label{eq:Filterpsi}
\end{align}
Therefore,
\begin{align}
q_\Lambda&(\Lambda_k)=\prod_{i=1}^{n_y}\G\left([\Lambda_k]_{ii};\tfrac{\nu_i}{2}+1,\tfrac{\nu_i+\Psi_{ii}}{2}\right).
\end{align}
In the derivations of Subection A, $\Lambda_{k|K} \!\triangleq\! \E_{q_\Lambda}[\Lambda_k]$ is required. $\Lambda_{k|K}$ is a diagonal matrix with the diagonal elements
\begin{align}
[\Lambda_{k|K}]_{ii}=\tfrac{\nu_i+2}{\nu_i+\Psi_{ii}}.
\end{align}

In the model \eqref{eq:measmod_mvskewt} with multivariate skew-$t$-distributed measurement noise $\Lambda_k$ is of the form $\lambda_k \cdot I_{n_y}$. There, $\lambda_k$ is a scalar random variable, and there is just one degrees-of-freedom parameter $\nu$, as given in \eqref{eq:lambda_mvst}. Therefore,
\begin{align}
&\log q_\Lambda(\lambda_k) \nonumber\\
=& -\frac{1}{2}\E_{q_{xu}}[\tr\{(y_k\!-\!Cx_k\!-\!\Delta u_k)(y_k\!-\!Cx_k\!-\!\Delta u_k)^\t R^{-1}\lambda_k\}] \nonumber\\
& - \frac{1}{2}\E_{q_{xu}}[\tr\{ u_k  u_k^\t\lambda_k\}] +\tfrac{\nu+2n_y-1}{2}\log\lambda_k - \frac{\nu}{2} \lambda_k +c\\
=& \tfrac{\nu+2n_y-1}{2} \log\lambda_k - \tfrac{\nu+\tr\{\Psi\}}{2}\lambda_k ,
\end{align}
where $\Psi$ is given in \eqref{eq:Filterpsi}. Thus,
\begin{align}
q_\Lambda(\lambda_k) &= \G\left(\lambda_k;\tfrac{\nu+2n_y}{2},\tfrac{\nu+\tr\{\Psi\}}{2}\right).
\end{align}
so the required expectation is
\begin{align}
\Lambda_{k|K} = \tfrac{\nu+2n_y}{\nu+\tr\{\Psi\}} \cdot \eye_{n_y} .
\end{align}

\allowdisplaybreaks
\section{Derivation for the Fisher information of MVST} \label{sec:crlb_derivations}

\newcommand{\half}{\frac{1}{2}}
\newcommand{\thalf}{\frac{\t}{2}}
\newcommand{\tpdf}{\mathrm{t}}
\newcommand{\tcdf}{\mathrm{T}}
\newcommand{\zz}{r}
\newcommand{\sqrtr}{\sqrt{\tfrac{\nu+n_y}{\nu+\zz^\t \zz}}}

Consider the multivariate skew-$t$ measurement model $y|x \sim \mathrm{MVST}(Cx,R,\Delta,\nu)$, where $C\in\mathbb{R}^{n_y\times n_x}$, $R\in\mathbb{R}^{n_y \times n_y}$, $\Delta \in \mathbb{R}^{n_y \times n_y}$, and $\nu\in\mathbb{R}^+$. The logarithm of the PDF of $y|x$ is
\begin{align} \label{eq:lhoodAppB}
\log p(y|&x) = \log(2^{n_y}/\det(\Omega)^\half) + \log\tpdf(\zz; 0,\eye_{n_y},\nu) \nonumber\\
&+ \log\tcdf(\Delta^\t \Omega^{-\thalf} \zz \sqrtr; 0,L, \nu+n_y) ,
\end{align}
where $r=\Omega^{-\frac{1}{2}}(y-Cx)$ is a function of $x$ and $y$, $\Omega=R+\Delta\Delta^\t$, $L= I_{n_y}-\Delta^\t\Omega^{-1}\Delta$, and $\tpdf(\cdot; \mu,\Sigma,\nu)$ and $\tcdf(\cdot; \mu,\Sigma,\nu)$ denote the PDF and CDF of the scaled non-central multivariate $t$-distribution with $\nu$ degrees of freedom. $A^\half$ is a square-root matrix such that $A^\half(A^\half)^\t=A$, $A^{-\half}\triangleq(A^\half)^{-1}$, and $A^{-\thalf}\triangleq((A^{\half})^{-1})^\t$.

The Hessian matrix of the term $\log\tpdf(\zz;0,\eye_{n_y},\nu)$ is derived in \cite{piche2016b}, and it is
\begin{align}
&\tfrac{\d^2}{\d x^2} \log\tpdf(\zz;0,I_{n_y},\nu) \nonumber\\
=& \tfrac{\nu+n_y}{\nu} C^\t \Omega^{-\thalf} \left( -\tfrac{1}{1+\tfrac{1}{\nu}\zz^\t \zz} I_{n_y} + \tfrac{2/\nu}{(1+\tfrac{1}{\nu}\zz^\t \zz)^2} \zz\zz^\t \right) \Omega^{-\half} C \\
=& \tfrac{\nu+n_y}{\nu+\zz^\t \zz} C^\t \Omega^{-\thalf} \left( -I_{n_y} +\tfrac{2}{\nu+\zz^\t \zz} \zz\zz^\t \right) \Omega^{-\half} C
\end{align}
The second term in \eqref{eq:lhoodAppB} can be differentiated twice using the chain rule $\tfrac{\d^2 \log(f)}{\d x^2} = \tfrac{1}{f} \tfrac{\d^2 f}{\d x^2} - \tfrac{1}{f^2} \left( \tfrac{\d f}{\d x}\right)^\t \! \tfrac{\d f}{\d x}$, which gives
\begin{align}
& \tfrac{\d^2}{\d x^2}  \log\tcdf(\Delta^\t \Omega^{-\thalf} \zz \sqrtr; 0,L,\nu+n_y) \nonumber\\
=& \big( \tcdf(\Delta^\t \Omega^{-\thalf} \zz \sqrtr; 0,L,\nu+n_y) \big)^{-1} g(\zz) \nonumber\\
&- \big( \tcdf(\Delta^\t \Omega^{-\thalf} \zz \sqrtr; 0,L,\nu+n_y) \big)^{-2} D_r^\t P_r^\t P_r D_r , \label{eq:fisher_integrand}
\end{align}
where the function $g$ is antisymmetric because it is the second derivative of a function that is antisymmetric up to an additive constant,
\begin{equation}
P_r=\tfrac{\d}{\d u} \tcdf(u;0,L,\nu+n_y)\Bigr|_{u=\Delta^\t \Omega^{-\thalf} \zz\sqrtr},
\end{equation}
and
\begin{align}
D_r =& \tfrac{\d}{\d x} \Delta^\t \Omega^{-\thalf} \zz\sqrtr \\
=& \sqrtr \Delta^\t \Omega^{-\thalf} (\tfrac{1}{\nu+\zz^\t \zz} \zz\zz^\t - I_{n_y}) \Omega^{-\half} C .
\end{align}
Because the function $g$ is antisymmetric, $\int g(r) p(r) \d y = 0$ for any symmetric function $p$ for which the integral exists.

We now outline the proof of integrability of certain functions to show that the CRLB exists and fulfils the regularity conditions given in \cite[Ch.\ 2.4]{vantrees}. The integral $\int g(r)\, \tpdf(r;0,1,\nu) \d y$ exists because the terms of $g$ are products of positive powers of rational expressions where the denominator is of a higher degree than the nominator and derivatives of $\mathrm{T}(u;0,1,\nu+n_y)$ evaluated at $\Delta^\t \Omega^{-\thalf} \zz\sqrtr$, which is a bounded continuous function of $y$. The integral
\[
\begin{split}
\int &\big( \tcdf(\Delta^\t \Omega^{-\thalf} \zz \sqrtr; 0,L,\nu+n_y) \big)^{-1} D_r^\t P_r^\t P_r D_r  \\
&\times \tfrac{2}{\det(\Omega)^{\half}} \tpdf(\zz; 0,\eye_{n_y},\nu) \d y
\end{split}
\]
also exists because $\big( \tcdf(\Delta^\t \Omega^{-\thalf} \zz \sqrtr; 0,L,\nu+n_y) \big)^{-1}$ and $P_r$ are bounded and continuous and $D_r$ is a positive power of a rational expression where the denominator is of a higher degree than the nominator. Similar arguments show the integrability of the first and second derivative of the likelihood $p(y|x)$, which guarantees that the regularity conditions of the CRLB are satisfied.

Thus, the expectation of \eqref{eq:fisher_integrand} is
\begin{align}
&\E_{p(y|x)}\left[\tfrac{\d^2}{\d x^2}  \log\tcdf(\Delta^\t \Omega^{-\thalf} \zz \sqrtr; 0,L,\nu+n_y) \right] \nonumber\\
=& \int g(\zz)\, \tfrac{2}{\det(\Omega)^{\half}} \tpdf(\zz; 0,\eye_{n_y},\nu) \d y - \int \tfrac{2}{\det(\Omega)^{\half}} \tpdf(\zz; 0,\eye_{n_y},\nu) \nonumber\\
& \hspace{-1.8mm}\times\! \big( \tcdf(\Delta^\t \Omega^{-\thalf} \zz \sqrtr; 0,L,\nu+n_y) \big)^{-1}\! D_r^\t P_r^\t P_r D_r \mathrm{d} y \\
=& \int 2g(\zz) \tpdf(\zz; 0,\eye_{n_y},\nu) \d r - \int 2 \,\tpdf(\zz; 0,\eye_{n_y},\nu) \nonumber\\
& \hspace{-1.8mm}\times \! \big( \tcdf(\Delta^\t \Omega^{-\thalf} \zz \sqrtr; 0,L,\nu+n_y) \big)^{-1} D_r^\t P_r^\t P_r D_r \d r \\
=&-\!  \E_{p(\zz|x)} \left[ \big( \tcdf(\Theta^\t \zz \sqrtr; 0,L,\nu+n_y) \big)^{-2} D_r^\t P_r^\t P_r D_r  \right]\!,
\end{align}
where $\Theta\!=\!\Omega^{-\half}\Delta$, and $\zz|x\!\sim\!{\mathrm{MVST}(0,I_{n_y}\!-\!\Theta\Theta^\t,\Theta,\nu)}$ because $z\!\sim\!\mathrm{MVST}(\mu,R,\Delta,\nu)$ implies $Az\!\sim\!\mathrm{MVST}(A\mu,ARA^\t,A\Delta,\nu)$. This gives
\begin{align}
\E_{p(y|x)}&\left[\tfrac{\d^2}{\d x^2}  \log\tcdf(\Delta^\t \Omega^{-\thalf} \zz \sqrtr;0,L,\nu+n_y) \right] \nonumber\\
=& -C^\t \Omega^{-\thalf} \E_{p(\zz|x)}\left[ \tfrac{\nu+n_y}{\nu+\zz^\t \zz} \widetilde{R}_\zz \widetilde{R}_\zz^\t \right] \Omega^{-\half} C ,
\end{align}
where
\begin{align} \label{eq:tildeR_def}
\widetilde{R}_\zz =& \big( \tcdf(\Theta^\t \zz \sqrtr; 0,L,\nu+n_y) \big)^{-1} (I_{n_y}-\tfrac{1}{\nu+\zz^\t \zz} \zz\zz^\t) \Theta\nonumber\\
&\times 
 \left(\tfrac{\d}{\d u} \tcdf(u; 0,L,\nu+n_y)\Bigr|_{u=\Theta^\t \zz\sqrtr} \right)^\t,
\end{align}
where $L\!=\!\eye_{n_y}-\Theta^\t \Theta$. Thus, the Fisher information for the measurement model $y|x\sim\mathrm{MVST}(Cx,R,\Delta,\nu)$ is
\begin{align}
&\mathcal{I}(x) = \E_{p(y|x)} \left[ -\tfrac{\d^2}{\d x^2} \log p(y|x) \right] \\
=& C^\t \Omega^{-\thalf} \!\! \E_{p(r|x)} \left[ \tfrac{\nu+n_y}{\nu+r^\t r} (I_{n_y} \!-\! \tfrac{2}{(\nu+\zz^\t \zz)^2} \zz\zz^\t \!+\! \widetilde{R}_\zz\widetilde{R}_\zz^\t) \right] \Omega^{-\half} C ,
\end{align}
where $r|x\sim\mathrm{MVST}(0,I_{n_y}\!-\!\Theta\Theta^\t,\Theta,\nu)$, $\Theta = \Omega^{-\half}\Delta$, $\Omega=R+\Delta\Delta^\t$, and $\widetilde{R}_\zz$ is defined in \eqref{eq:tildeR_def}.

\end{document}